\documentclass[11pt,a4paper]{article}

\usepackage{a4wide}
\usepackage{amsmath}
\usepackage{url}
\usepackage{epsfig}
\usepackage{psfrag}
\usepackage{amssymb}
\usepackage{wrapfig}
\usepackage{graphicx}
\usepackage{tikz}
\usepackage{verbatim}
\usepackage{enumerate}
\usepackage{algorithmic}
\usepackage{algorithm}
\usepackage{hyperref}

\newcommand{\calC}{{\cal C}}

\def\capa(#1){C_{#1}}
\newcommand{\argmax}{\mathop \mathrm{argmax}}

\newcommand{\dist}{\mathrm{dist}} 

\newcommand{\source}{s_0}
\newcommand{\sink}{s_1}

\newcommand{\donotshow}[1]{}

\newcommand{\ignore}[1]{}

\newcommand{\assign}{\mathbin{\raisebox{0.05ex}{\mbox{\rm :}}\!\!=}}
\newcommand{\define}{\mathbin{\raisebox{0.05ex}{\mbox{\rm :}}\!\!=}}

\newcommand{\RR}{\mathbb{R}}
\providecommand{\R}{\mathbb{R}}

\newcommand {\abs}[1] {| #1 |}

\newlength{\setspacing}
\providecommand{\sset}[1]{\{ #1 \}}
\newcommand{\set}[2]{ \left\{ #1 \mbox{ ; } #2 \right\}  }

\newcommand{\uedge}[1]{\sset{#1}}

\newcommand{\mbegin}{\{\ \ }
\newcommand{\mend}{\}}

\newlength{\mleftindent}
\newlength{\mindent}
\newlength{\mboxwidth}
\newcommand{\mincrement}{\addtolength{\mboxwidth}{-\mindent}}
\newcommand{\mdecrement}{\addtolength{\mboxwidth}{\mindent}}

\newlength{\preprogramskip}
\newlength{\postprogramskip}
\newlength{\mexpwidth}
\newlength{\mexpindent}
\newcommand{\indentafterkeyword}{\hspace*{0.5em}}

\newcommand{\mslifelse}[3]  
{\setlength{\mexpwidth}{\mboxwidth}%
\settowidth{\mexpindent}{{\bf if\indentafterkeyword}}%
\addtolength{\mexpwidth}{-\mexpindent}%
{\bf if\indentafterkeyword}\parbox[t]{\mexpwidth}{#1}\\
\mincrement \mbegin \parbox[t]{\mboxwidth}{#2 \mend} \mdecrement \\
{\bf else} \\
\mincrement \mbegin \parbox[t]{\mboxwidth}{#3}\\
\mend \mdecrement
}

{\vspace{-0.3em}\begin{itemize}
\setlength{\itemsep}{0.2\itemsep}%
\setlength{\parskip}{0.2\parskip}%
\setlength{\topsep}{0.0\topsep}%
}
{\end{itemize}}

{\begin{itemize}
\setlength{\itemsep}{0.2\itemsep}%
\setlength{\parskip}{0.2\parskip}%
\setlength{\topsep}{0.0\topsep}%
}
{\end{itemize}}

{\begin{itemize}
\setlength{\itemsep}{1.5\itemsep}%
\setlength{\parskip}{1.5\parskip}%
\setlength{\topsep}{0.0\topsep}%
}
{\end{itemize}}

{\begin{itemize}
\setlength{\itemsep}{1.5\itemsep}%
\setlength{\parskip}{1.5\parskip}%
\par \smallskip
}
{\end{itemize}}

{\begin{itemize}
\setlength{\itemsep}{1.5\itemsep}%
\setlength{\parskip}{1.5\parskip}%
\par \medskip
}
{\end{itemize}}

{\vspace{-0.3em}\begin{description}
\setlength{\itemsep}{0.2\itemsep}%
\setlength{\parskip}{0.2\parskip}%
\setlength{\topsep}{0.0\topsep}%
}
{\end{description}}

{\begin{description}
\setlength{\itemsep}{0.2\itemsep}%
\setlength{\parskip}{0.2\parskip}%
\setlength{\topsep}{0.0\topsep}%
}
{\end{description}}

{\vspace{-0.3em}\begin{enumerate}
\setlength{\itemsep}{0.2\itemsep}%
\setlength{\parskip}{0.2\parskip}%
\setlength{\topsep}{0.0\topsep}%
}
{\end{enumerate}}

{\begin{enumerate}
\setlength{\itemsep}{0.2\itemsep}%
\setlength{\parskip}{0.2\parskip}%
\setlength{\topsep}{0.0\topsep}%
}
{\end{enumerate}}

\newlength{\proofpostskipamount}\newlength{\proofpreskipamount}
\newenvironment{proof}%
               {\par\vspace{\proofpreskipamount}\noindent{\bf Proof:}\hspace{0.5em}}
               {\nopagebreak%
                \strut\nopagebreak%
                \hspace{\fill}\qed\par\vspace{\proofpostskipamount}\noindent}

               {\par\vspace{0.5ex}\noindent{\bf Proof #1:}\hspace{0.5em}}%
               {\nopagebreak%
                \strut\nopagebreak%
                \hspace{\fill}\qed\par\medskip\noindent}

\newtheorem{theorem}{Theorem}
\newtheorem{lemma}{Lemma}

\newtheorem{corollary}[theorem]{Corollary}
\newtheorem{definition}{Definition}
\providecommand{\qed}{\rule[-0.2ex]{0.3em}{1.4ex}}

\newlength{\mydefwidth}
\newlength{\mytextwidth}

\newcommand{\myurl}[1]{{\footnotesize \url{#1}}}

\newcommand{\Real}{\R}

\newcommand{\firsttheorem}[2]{\medskip\par\noindent{\bf Theorem #1} {\it
#2}\medskip}
\newcommand{\firstlemma}[2]{\medskip\par\noindent{\bf Lemma #1} {\it #2}\medskip}

\renewcommand{\source}{\ensuremath{s_0}}
\renewcommand{\sink}{\ensuremath{s_1}}
\newcommand{\Lmin}{\ensuremath{L_{\mathrm{min}}}}
\newcommand{\Lmax}{\ensuremath{L_{\mathrm{max}}}}
\newcommand{\diverg}{b}
 
\newcommand{\Eorient}{\overrightarrow{E}} \newcommand{\CB}{\mathit{F}}
\def\mycapa_#1(#2){C_{#2}(#1)}
\newcommand{\negskip}{\vspace{-2em}\par}\newcommand{\Lnegskip}{\vspace{-3em}\par}
\newcommand{\E}{\cal E}
\newcommand{\Estar}{\E^*}
\newcommand{\sptree}{\mathrm{Sp}}
\newcommand{\sign}{\mathrm{sign}}

\begin{document}

\title{Physarum Can Compute Shortest Paths\thanks{An extended abstract of this
paper appears in SODA (ACM-SIAM Symposium on Discrete Algorithms) 2012.}}
\author{Vincenzo Bonifaci\thanks{Istituto di Analisi dei Sistemi ed Informatica
``Antonio Ruberti'' -- CNR, Rome, Italy. Most of the work was done at
the MPI for Informatics, Saarbr\"{u}cken, Germany. Email: \texttt{bonifaci@mpi-inf.mpg.de}.} 
\and 
Kurt Mehlhorn\thanks{MPI for Informatics, Saarbr\"{u}cken, Germany. Email: \texttt{mehlhorn@mpi-inf.mpg.de}.} 
\and 
Girish Varma\thanks{Tata Institute of Fundamental Research, Mumbai, India. Most
of the work was done at the MPI for Informatics, Saarbr\"{u}cken, Germany. Email: \texttt{girish@tcs.tifr.res.in}.} 
}

\maketitle

\begin{abstract}  
\emph{Physarum Polycephalum} is a slime mold that is apparently able to solve
shortest path problems. 
A mathematical model has been proposed by biologists to describe the feedback
mechanism used by the slime mold to adapt its tubular
channels while foraging two food sources $s_0$ and $s_1$. We prove
that, under this model, the mass of the mold will eventually converge
to the shortest $s_0$-$s_1$ path of the network that the mold lies on,
independently of the structure of the network or of the initial mass
distribution.  

This matches the experimental observations by the biologists and can be seen as
an example of a ``natural algorithm'', that is, an algorithm developed by
evolution over millions of years.   
\end{abstract}

\tableofcontents
\thispagestyle{empty}

\newpage

\section{Introduction}

\begin{figure}[t]
\begin{center}
\includegraphics[width=0.45\textwidth]{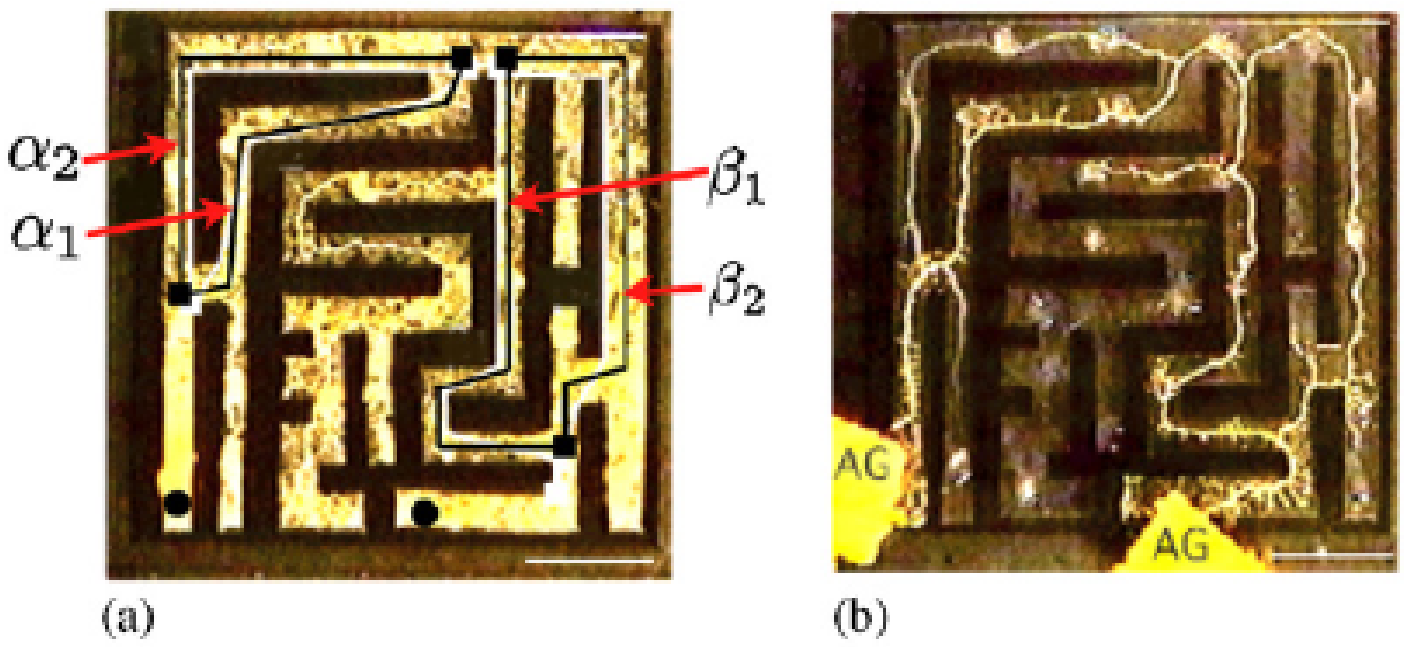}
\hspace{0.3cm}
\includegraphics[width=0.43\textwidth]{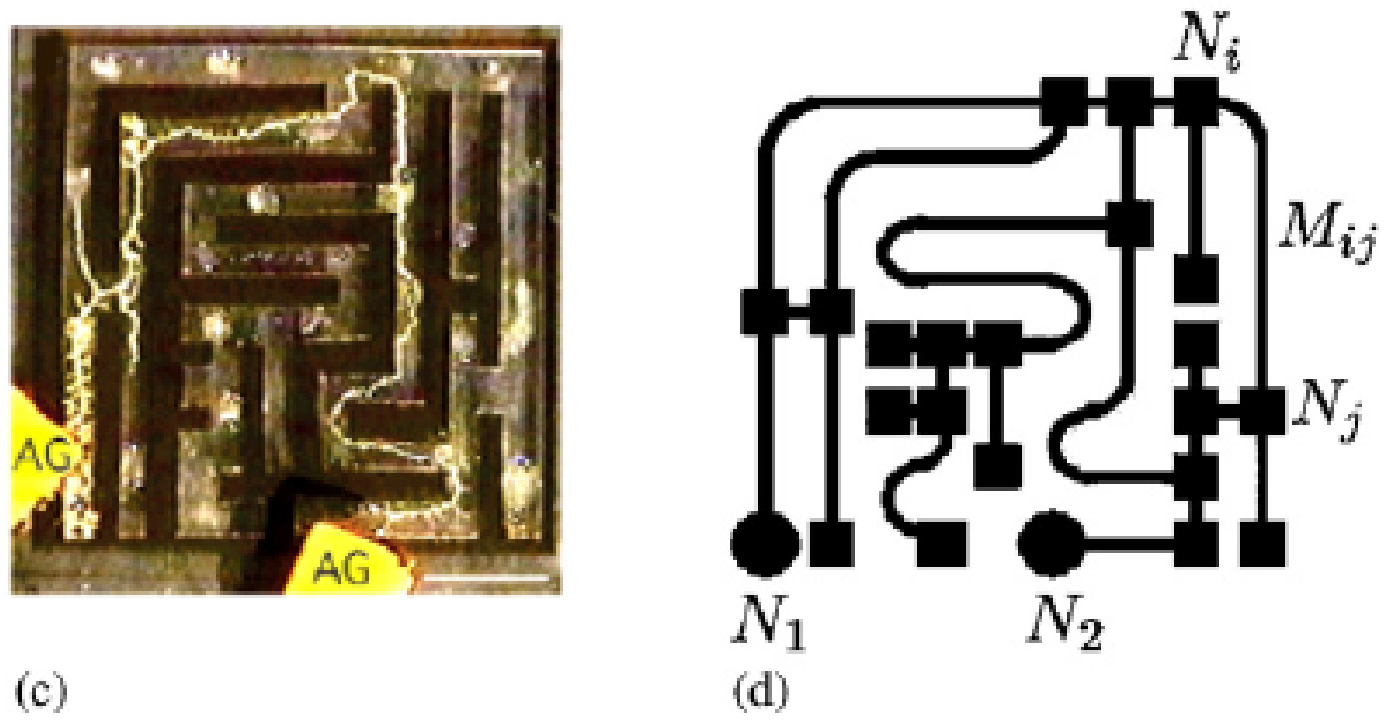}
\end{center}
\caption{\label{fig:maze} The experiment in~\cite{Nakagaki-Yamada-Toth}
(reprinted from there): (a) shows the maze uniformly covered by Physarum; the
yellow color indicates the presence of Physarum. Food (oatmeal) is provided at the
locations labelled AG. After a while, the mold retracts to the shortest path
connecting the food sources as shown in (b) and (c). (d) shows the underlying
abstract graph. The video \cite{youtube:slime:mold:video} shows the
experiment. 
}
\end{figure}

\emph{Physarum Polycephalum} is a slime mold that is apparently able to solve
shortest path problems. Nakagaki, Yamada, and
T\'{o}th~\cite{Nakagaki-Yamada-Toth} report on the following experiment, see
Figure~\ref{fig:maze}: They built a maze, covered it with pieces of Physarum (the
slime can be cut into pieces that will reunite if brought into vicinity), and
then fed the slime with oatmeal at two locations. After a few hours, the slime
retracted to a path that follows the shortest path connecting the food sources
in the maze. 
The authors report that they repeated the experiment with
different mazes; in all experiments, Physarum retracted to the shortest
path. There are several videos available on the web that show the mold in
action \cite{youtube:slime:mold:video}.   

Paper~\cite{Tero-Kobayashi-Nakagaki} proposes a mathematical model for the
behavior of the slime and argues extensively that the model is adequate. We
will not repeat the discussion here but only define the model. Physarum is
modeled as an electrical network with time varying resistors. We have a simple
undirected graph $G = (N,E)$ with distinguished nodes $\source$ and $\sink$, which
model the food sources. Each edge $e \in E$ has a positive length $L_e$ and
a positive diameter $D_e(t)$; $L_e$ is fixed, but $D_e(t)$ is a function of
time. The resistance $R_e(t)$ of $e$ is $R_e(t) = L_e/D_e(t)$. We force a
current of value 1 from $\source$ to $\sink$. Let $Q_e(t)$ be the resulting
current over edge $e = (u,v)$, where $(u,v)$ is an arbitrary orientation of
$e$. The diameter of any edge $e$ evolves according to the equation 
\begin{equation} 
\dot{D_e}(t) = | Q_e(t) | - D_e(t) , \label{dynamics} 
\end{equation}
where $\dot{D_e}$ is the derivative of $D_e$ with respect to time. In
equilibrium ($\dot{D_e} = 0$ for all $e$), the flow through any edge is equal
to its diameter. In non-equilibrium, the diameter grows or shrinks if the
absolute value of the flow is larger or smaller than the diameter, respectively. In the
sequel, we will mostly drop the argument $t$ as is customary in the treatment
of dynamical systems.  

The model is readily turned into a computer simulation. In an electrical
network, every vertex $v$ has a potential $p_v$; $p_v$ is a function of
time. We may fix $p_{\sink}$ to zero. 
For an edge $e = (u,v)$, the flow across $e$ is given by $(p_u -
p_v)/R_e$. We have flow conservation in every vertex except for $\source$ and
$\sink$; we inject one unit at $\source$ and remove one unit at $\sink$. Thus,
\begin{equation}
b_v = \sum_{u \in \delta(v)} \frac{p_v - p_u}{R_{uv}}, \label{potentials}
\end{equation}
where $\delta(v)$ is the set of neighbors of $v$ and $b_{\source} = 1$, $b_{\sink}
= -1$, and $b_v = 0$ otherwise. The node potentials can be computed by solving a
linear system (either directly or iteratively).  
Tero et
al.~\cite{Tero-Kobayashi-Nakagaki} were the first to perform 
simulations of the model. They report that the network always converges to the shortest
$\source$-$\sink$ path, i.e., the diameters of the edges on the shortest path
converge to one, and the diameters on the edges outside the shortest path
converge to zero. This holds true for any initial condition and assumes the 
uniqueness of the shortest path. 

Miyaji and Ohnishi~\cite{Miyaji-Ohnish07,Miyaji-Ohnishi} initiated the analytical
investigation of the model. They argued convergence against the shortest path
if $G$ is a planar graph and $\source$ and $\sink$ lie on the same face in some
embedding of $G$. 

Our main result is a convergence proof for all graphs. For a network $G =
(V,E,\source,\sink,L)$, where $(L_e)_{e \in E}$ is a positive length function
on the edges of $G$, we use $G_0 = (V,E_0)$ to denote the subgraph of all
shortest source-sink paths, $L^*$ to denote the length of a shortest
source-sink path, and $\Estar$ to denote the set of all source-sink flows of
value one in $G_0$. If we define the cost of flow $Q$ as $\sum_e L_e Q_e$, then
$\Estar$ is the set of minimum cost source-sink flows of value one. If the
shortest source-sink path is unique, $\Estar$ is a singleton. The 
dynamics are \emph{attracted} by a set $A \subseteq \R^E$ if the distance (measured in any
$L_p$-norm) between $D(t)$ and $A$ converges to zero.

\firsttheorem{[Theorem~\ref{thm:main-apx} in
Section~\ref{sec:Convergence-General}]}{Let $G = (V,E,\source,\sink,L)$ be an undirected network with positive length
function $(L_e)_{e \in E}$. Let $D_e(0) > 0$ be the diameter of edge $e$ at time
zero. 
The dynamics (\ref{dynamics}) are attracted to $\Estar$. If the
shortest source-sink path is unique, the dynamics converge to the flow of value
one along the shortest source-sink path.}

We conjecture that the dynamics converge to an element of $\Estar$ but only
show attraction to $\Estar$. A key part of our proof is to show that the function
\begin{equation}\label{eq: defintion of V}
V = \frac{1}{\min_{S \in \calC} \capa(S)} \sum_{e \in E} L_e D_e + (\capa(\{ \source \})-1)^2 \end{equation}
decreases along all trajectories that start in a non-equilibrium
configuration. Here, $\calC$ is the set of all $\source$-$\sink$ cuts, i.e.,
the set of all $S \subseteq N$ with $\source \in S$ and $\sink \not\in S$;
$\capa(S) = \sum_{e \in \delta(S)} D_e$ is the capacity of the cut $S$ when the
capacity of edge $e$ is set to $D_e$; and $\min_{S \in \calC} \capa(S)$ (also
abbreviated by $C$) is the capacity of the minimum cut. The first term in the
definition of $V$ is the normalized hardware cost; for any edge, the product of
its length and its diameter may be interpreted as the hardware cost of the
edge; the normalization is by the capacity of the minimum cut. We will show
that the first term
decreases except when the maximum flow $F$ in
the network with capacities $D_e$ is unique, and moreover, 
$\abs{Q_e} = \abs{F_e}/C$ for all $e$. The second term decreases as long as the
capacity of the cut defined by $\source$ is different from 1. We show that the
capacity of the minimum cut converges to one and that the derivative of
$V$ is upper bounded by $- \sum_e (\Lmin/4) (D_e/C - \abs{Q_e})^2$, where $\Lmin$
is the minimum length of any edge. Since $V$ is
non-negative, this will allow us to conclude 
that $\abs{D_e - \abs{Q_e}}$ converges
to zero for all $e$. In the next step, we show that the potential difference
$\Delta = p_{\source} - p_{\sink}$ between source and sink converges to the
length $L^*$ of a shortest-source sink path. We use this to conclude that $D_e$
and $Q_e$ converge to zero for any edge $e \not\in E_0$. Finally, we show that
the dynamics are attracted by $\Estar$.  

We
found the function $V$ by analytical investigation of a network of
parallel links (see Section~\ref{sec: parallel links}), extensive
computer simulations, and guessing.  
Functions decreasing along all trajectories are called Lyapunov functions in
dynamical systems theory~\cite{Hirsch-Smale}. The fact that the right-hand side
of system (\ref{dynamics}) is not continuously differentiable and that the function $V$ is not
differentiable everywhere introduces some technical difficulties. \smallskip

The direction of the flow across an edge depends on the initial conditions and
time. We do not know whether flow directions can change infinitely
often or whether they become ultimately fixed. Under the assumption that flow
directions stabilize, we can characterize the (late stages of the) convergence
process. An edge $e = \uedge{u,v}$ becomes \emph{horizontal} if $\lim_{t \to \infty} \abs{p_u - p_v} =
0$, and it becomes \emph{directed} from $u$ to $v$ (directed from $v$ to $u$)
if $p_u > p_v$ for all large $t$ ($p_v > p_u$ for all large $t$). An edge
\emph{stabilizes} if it either becomes horizontal or directed, and a network
\emph{stabilizes} if all its edges stabilize. If a network stabilizes, we
partition its edges into a set $E_h$ of horizontal edges and a set $\Eorient$ of
directed edges. If $\uedge{u,v}$ becomes directed from $u$ to $v$, then 
$(u,v) \in \Eorient$. 

We introduce the notion of a \emph{decay rate}. Let $r \le 0$. 
A \emph{quantity $D(t)$ decays with rate at least
$r$} if for every $\varepsilon > 0$ there is a constant $A$ such that 
$\ln D(t) \le A + (r + \varepsilon) t$ for all $t$. 
A \emph{quantity $D(t)$ decays with rate at most $r$} if for every
$\varepsilon > 0$ there is a constant $a$
such that 
$ \ln D(t) \ge  a + (r - \varepsilon) t$ for all $t$. 
A {quantity $D(t)$ decays with rate $r$} if it decays with rate at least and at
most $r$. 

\firstlemma{[Lemma~\ref{lem: decay for Eh} in Section~\ref{sec: stable flow directions}]}{For $e \in E_h$, $D_e$ decays with rate
$- 1$ and $\abs{Q_e}$ decays with rate at least $-1$.}

\ignore{\begin{figure}[t]
\begin{center}
\psfrag{e1}{$e_1$}\psfrag{e2}{$e_2$}\psfrag{e3}{$e_3$}\psfrag{e4}{$e_4$}\psfrag{e5}{$e_5$}
\psfrag{e6}{$e_5$}\psfrag{s0}{$s_0$}\psfrag{s1}{$s_1$}\psfrag{u}{$u$}\psfrag{v}{$v$}\psfrag{w}{$w$}
\includegraphics[width=0.4\textwidth]{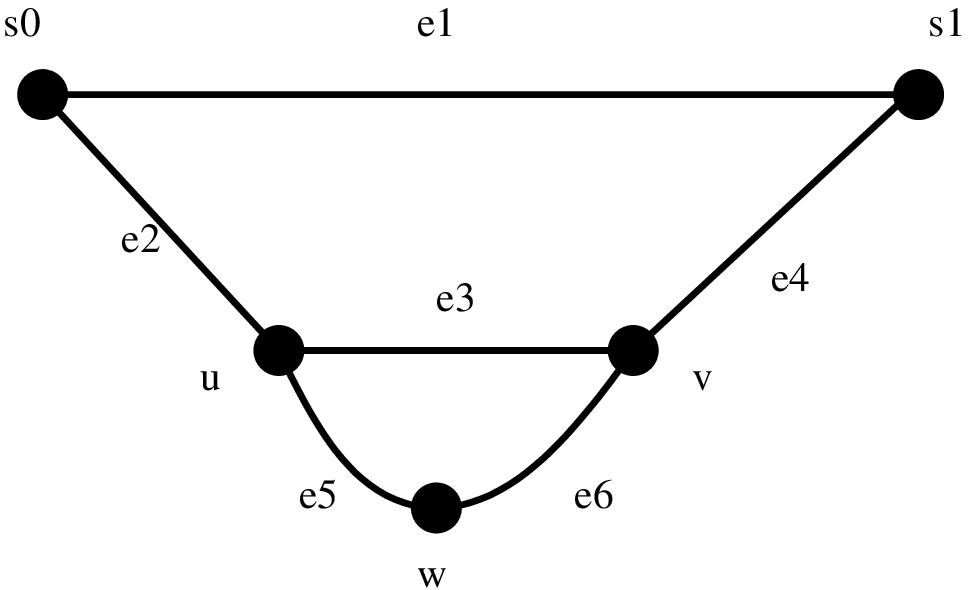}
\end{center}
\caption{%
All edges are assumed to have length
1; $P_0 = (e_1)$, $P_1 = (e_2,e_3,e_4)$, $P_2 = (e_5,e_6)$, $p^*_{s_0} = 1$,  $p^*_{s_1} = 0$, 
$p^*_v = 1/3$, $p^*_u = 2/3$, $p^*_w = 1/2$, $f(P_1) = 1/3$, and $f(P_2) =
1/6$. \protect\\
The path $(e_2,e_5,e_6,e_4)$ has $f$-value $1/4$.}
\end{figure}}

\begin{figure}[t]
\begin{center}
\begin{tabular}{c c c}
\psfrag{e1}{$e_1$}\psfrag{e2}{$e_2$}\psfrag{e3}{$e_3$}\psfrag{e4}{$e_4$}\psfrag{e5}{$e_5$}
\psfrag{e6}{$e_6$}\psfrag{s0}{$s_0$}\psfrag{s1}{$s_1$}\psfrag{u}{$u$}\psfrag{v}{$v$}\psfrag{w}{$w$}
\includegraphics[width=0.4\textwidth]{pathdecomposition.eps} &
\qquad\qquad & \begin{tikzpicture}[-,thick,node distance=2 cm] 
\tikzstyle{reg}=[circle,draw=black,fill=black,scale=0.9] 
\node (v)  [reg,label={left:$s_0$}] {}; 
\node (s)  [reg,above right of=v,label={above:$u$}] {}; 
\node (t)  [reg,below right of=v,label={below:$v$}] {};
\node (w)  [reg,above right of=t,label={right:$s_1$}] {};
\path (v)  edge node [label={above left:$a$}] {} (s);
\path (v)  edge node [label={below left:$b$}] {} (t); 
\path (s)  edge node [label={above right:$c$}] {} (w); 
\path (s)  edge node [label={left:$e$}] {} (t); 
\path (t)  edge node [label={below right:$d$}] {} (w);
\end{tikzpicture}\\
(a)  &   & (b) 
\end{tabular}
\end{center}
\caption{\label{fig: path decomposition-abstract} Part (a) illustrates the path
decomposition. All edges are assumed to have length
1; $P_0 = (e_1)$, $P_1 = (e_2,e_3,e_4)$, $P_2 = (e_5,e_6)$, $p^*_{s_0} = 1$,  $p^*_{s_1} = 0$, 
$p^*_v = 1/3$, $p^*_u = 2/3$, $p^*_w = 1/2$, $f(P_1) = 1/3$, and $f(P_2) =
1/6$. \protect\\
Part (b) shows the Wheatstone graph. The direction of the flow on edge $\uedge{u,v}$ may change over time; the flow on all other edges is always from left to right.}
\end{figure}

We define a decomposition of $G$ into paths $P_0$ to $P_k$, an orientation of
these paths, a slope $f(P_i)$ for each $P_i$, a vertex
labelling $p^*$, and an edge labelling $r$. 
$P_0$ is a\footnote{We assume that $P_0$ is unique.}
shortest $\source$-$\sink$ path in $G$, $f(P_0) = 1$, $r_e = f(P_0)-1$ for all $e \in P_0$,
and $p_v^* = \dist(v,\sink)$ for all $v \in P_0$, where $\dist(v,\sink)$ is the
shortest path distance from $v$ to $\sink$. 
For $1 \le i \le k$, we have\footnote{We assume that $P_i$ is unique except if $f(P_i) = 0$.} $P_i = \argmax_{P \in \cal P} f(P)$, 
where $\cal P$ is the set of all paths $P$ in $G$ with the following properties:
(1) the startpoint $a$ and the endpoint $b$ of $P$ lie on $P_0 \cup \ldots \cup
P_{i-1}$, $p^*_a \ge p_b^*$, and $f(P) = (p_a^* - p_b^*)/L(P)$; 
(2) no interior vertex of $P$ lies on $P_0 \cup \ldots \cup
P_{i-1}$; and
(3) no edge of $P$ belongs to $P_0 \cup \ldots \cup
P_{i-1}$.
If $p^*_a > p_b^*$, we direct $P_i$ from $a$ to $b$. If $p^*_a = p^*_b$, we
leave the edges in $P_i$ undirected. We set 
$r_e = f(P_i) - 1$ for all edges of $P_i$, and $p^*_v = p^*_b +
f(P_i)\, \dist_{P_i}(v,b)$  for every interior vertex $v$ of
$P_i$. Figure~\ref{fig: path decomposition-abstract}(a) illustrates the path 
decomposition. 

\firstlemma{[Lemma~\ref{lem: path decomposition} in Section~\ref{sec: stable
flow directions}]}{There is an
$i_0 \le k$ such that
\[   f(P_0) > f(P_1) >   \ldots >  f(P_{i_0}) > 0 = f(P_{i_0 + 1}) = \ldots =
f(P_k). \]}

\firsttheorem{[Theorem~\ref{thm: convergence for stable networks} in
Section~\ref{sec: stable flow directions}]}{If a network stabilizes, $\Eorient = \cup_{i \le i_0} E(P_i)$,
the orientation of any edge $e \in \Eorient$ agrees with the orientation 
induced by the path decomposition, and $E_h = \cup_{i > i_0} E(P_i)$. 
The potential of each node $v$ converges to $p^*_v$. The diameter of each edge $e
\in E \setminus P_0$ decays with rate $r_e$.}

\ignore{\begin{wrapfigure}[16]{r}{0.33\textwidth}
\begin{center}
\begin{tikzpicture}[-,thick,node distance=2 cm] 
\tikzstyle{reg}=[circle,draw=black] 
\node (v)  [reg] {$s_0$}; 
\node (s)  [reg,above right of=v] {$u$}; 
\node (t)  [reg,below right of=v] {$v$};
\node (w)  [reg,above right of=t] {$s_1$};
\path (v)  edge node [label={above left:$a$}] {} (s);
\path (v)  edge node [label={below left:$b$}] {} (t); 
\path (s)  edge node [label={above right:$c$}] {} (w); 
\path (s)  edge node [label={left:$e$}] {} (t); 
\path (t)  edge node [label={below right:$d$}] {} (w);
\end{tikzpicture}
\caption{The direction of the flow on edge $\uedge{u,v}$ may change over time; the flow on all other edges is always from left to right.}
\end{center} 
\end{wrapfigure}}


We cannot prove that flow directions stabilize in general. For
series-parallel graphs, flow directions trivially stabilize. The Wheatstone
graph, shown in
Figure~\ref{fig: path decomposition-abstract}(b), is the simplest graph, in
which flow directions may change over time. 

\firsttheorem{[Theorem~\ref{thm: Wheatstone} in Section~\ref{sec: Wheatstone}]}{The Wheatstone graph
stabilizes.}

The uncapacitated transportation problem generalizes the shortest
path problem. With each vertex $v$, a
supply/demand $b_v$ is associated. It is assumed that $\sum_v b_v = 0$. Nodes
with positive $b_v$ are called supply nodes, and nodes with negative $b_v$ are
called demand nodes. In the shortest path problem, exactly two vertices have
non-zero supply/demand. A feasible solution to the transportation problem is a
flow $f$ satisfying the mass balance constraints, i.e., for every vertex $v$,
$b_v$ is equal to the net flow out of $v$. The cost of a solution is $\sum_e
L_e f_e$. The Physarum solver for the transportation problem is as follows: At any
fixed time, the potentials are defined by (\ref{potentials}) and the currents
$(Q_e)_{e \in E}$ are derived from the potentials by Ohm's law. 
The dynamics evolve according to (\ref{dynamics}). The
equilibria, i.e., $\abs{Q_e} = D_e$ for all $e$, are precisely the flows with
the following equal-length property. Orient the edges in the direction of $Q$ and drop
the edges of flow zero. In the resulting graph, any two distinct directed paths with the same
source and sink have the same length. Let $\E$ be the set of equilibria. 

\firsttheorem{[Theorem~\ref{thm:main transportation} in Section~\ref{sec:transportation}]}{The
dynamics (\ref{dynamics}) are attracted to the set of equilibria $\E$. If any two
equilibria have distinct cost, the dynamics converge to an optimum solution of
the transportation problem.}

The convergence statement for the transportation problem is weaker than the
corresponding statement for the shortest path problem in two respects. There, we show attraction to
the set of equilibria of minimum cost (now only to the set of equilibria) and
convergence to the optimum solution if the optimum solution is unique (now only
if no two equilibria have the same cost). 

This paper is organized as follows: In Section~\ref{sec: related work}, we
discuss related work, and in Section~\ref{sec: discussion}, we put our results
into the context of natural algorithms and state open problems. The technical
part of the paper starts in Section~\ref{sec: parallel links}. We first treat a
network of parallel links; this situation is simple enough to allow an
analytical treatment. In Section~\ref{sec: preliminaries}, we review basic facts
about electrical networks and prove some simple facts about the dynamics of
Physarum. In Section~\ref{sec:Convergence-General}, we prove our main result,
the convergence for general graphs. In Section~\ref{sec: stable flow directions},
we prove exponential convergence under the assumption that flow directions
stabilize, and in Section~\ref{sec: Wheatstone}, we show that the Wheatstone
network stabilizes. Finally, in Section~\ref{sec:transportation}, we generalize
the convergence proof to the transportation problem.

\section{Related Work}\label{sec: related work}
Miyaji and Ohnishi~\cite{Miyaji-Ohnish07,Miyaji-Ohnishi}
initiated the analytical investigation of the model. They argued convergence
against the shortest path if $G$ is a planar graph and 
$\source$ and $\sink$ lie on the same face in some embedding of $G$. 
Ito et al.~\cite{Ito-Convergence-Physarum} study the dynamics 
(\ref{dynamics}) in a directed graph $G = (V,E)$; they do not claim that the model is justified on  
biological grounds. Each directed edge $e$ has a
diameter $D_e$. The node potentials are again defined by the equations
\[   b_v = \sum_{u \in \delta(v)} \frac{p_v - p_u}{R_{uv}} \quad\text{for all
$v \in V$}. \]
The summation on the right-hand side is over all neighbors $u$ of $v$; edge
directions do not matter in this equation. If there is an edge from $u$ to $v$
and an edge from $v$ to $u$, $u$ occurs twice in the summation, once for each
edge. 
The dynamics for the diameter of the directed edge $uv$ are then 
$\dot{D}_{uv} = Q_{uv} - D_{uv}$, where $Q_{uv} = D_{uv}(p_u -
p_v)/L_{uv}$. The dynamics of this model are very
different from the dynamics of the model studied in our paper. For example,
assume that there is an edge $vu$, no edge $uv$, and $p_u > p_v$ always. Then
$Q_{vu} < 0$ always and hence $D_{vu}$ will vanish at least with rate $-1$. 
The model is simpler to analyze than our model. Ito et al.~prove that the
directed model is able to solve transportation problems and that the $D_e$'s
converge exponentially to their limit values. 

\section{Discussion and Open Problems}\label{sec: discussion}
Physarum may be seen as an example of a natural computer, i.e., a computer
developed by evolution over millions of years. It can apparently do more than
compute shortest paths and solve transportation problems. 
In~\cite{Tero-Takagi-etal}, the computational capabilities of Physarum
are applied to network design, and it is shown in lab and computer experiments
that Physarum can compute approximate Steiner trees. No theoretical analysis is
available. The book~\cite{PhysarumBook} and the tutorial \cite{PhysarumTutorial} contain many illustrative examples of the computational power of this slime mold. 

Chazelle~\cite{Chazelle-NaturalAlgorithms} advocates the study of natural
algorithms; i.e., ``algorithms developed by evolution over millions of years'',
using computer science techniques. Traditionally, the analysis of such
algorithms belonged to the domain of biology, systems theory, and physics. Computer
science brings new tools. For example, in our analysis, we crucially use the
max-flow min-cut theorem.  
Natural algorithms can also give inspiration for the development of new
combinatorial algorithms. A good example is~\cite{Christiano-et-al}, where
electrical flows are essential for an approximation algorithm for undirected
network flow. 

We have only started the theoretical investigation of Physarum
computation, and so many interesting questions are open. We prove convergence for the
dynamics $\dot{D_e} = f(\abs{Q_e}) - D_e$, where $f$ is the identity
function. The biological literature also suggests the use of $f(x) =
x^\gamma/(1 + x^\gamma)$ for some parameter $\gamma$. Can one prove convergence for
other functions $f$? We prove that flow directions stabilize in the Wheatstone
graph. Do they stabilize in general? We prove, but only for stabilizing networks, that
the diameters of edges not on the shortest path converge to zero exponentially
for large $t$. What can be said about the initial stages of the process? 
The Physarum computation is fully distributed; node potentials depend only on
the potentials of the neighbors, currents are determined by potential
differences of edge endpoints, and the update rule for edge diameters is
local. Can the Physarum computation be used as the basis for an efficient
distributed shortest path algorithm? What other problems can be provably solved
with Physarum computations?

\section{Parallel Links}\label{sec: parallel links}

We discovered the Lyapunov function used in the proof of our main theorem
through experimentation. The experimentation was guided by the analysis of a
network of parallel links. In such a network, there are vertices $\source$ and
$\sink$ connected with $m$ edges of lengths $L_1 < L_2 < \ldots < L_m$. Let $D_i$ be the
diameter of the $i$-th link, and let $D = \sum_i D_i$. Let
$\Delta = p_{\source} - p_{\sink}$ be the potential difference between source and
sink. Then, $Q_i = \Delta/R_i = D_i \Delta/ L_i$. Since $\sum_i Q_i = 1$, we
have $\Delta = 1/\sum_i D_i/L_i$. 

\begin{lemma} The equilibrium points are precisely the single links. \end{lemma}
\begin{proof} In an equilibrium point, $Q_i = D_i$ for all $i$. Since $Q_i =
D_i \Delta/L_i$, this implies $\Delta = L_i$ whenever $Q_i \not= 0$. Thus, in an
equilibrium there is exactly one $i$ with $Q_i \not= 0$. Then, $Q_i =
1$. \end{proof}

\begin{lemma}\label{lem:d-sum} 
Let $D=\sum_i D_i$. Then, $D$ converges to 1. 
\end{lemma}
\begin{proof}
We have 
$\dot{D} = \sum_i \dot{D_i} = \sum_i Q_i - \sum_i D_i = 1 - D.$ 
The claim follows by directly solving the differential equation: $D(t) = 1 + (D(0)-1) \exp(-t)$.  
\end{proof}

For networks of parallel links, there are many Lyapunov functions. 

\begin{lemma}\label{lem:parallel links} Let $x_i = D_i/D$, and let $L$ be such
that $1/L = \sum_j x_j/L_j$. The quantities 
\[ \sum_{i \ge 2} D_i/D,\ \sum_i x_i L_i,\ L,\ \sum_i Q_i L_i,\ \Delta \sum_i
D_i L_i,\ \text{and } \sum_{i \ge 2} (L_i \ln D_i - L_1
\ln D_1)\] decrease
along all trajectories, starting in non-equilibrium points. \end{lemma}
\begin{proof}  Clearly, $\sum_j x_j = 1$ and 
$\Delta = L/D$.  The derivative $\dot{x_i}$ of $x_i$ computes as: 
\[ \dot{x_i} = \frac{\dot{D_i}D - D_i \dot{D}}{D^2} 
          = \frac{(D_i\Delta/L_i - D_i)D - D_i(1 - D)}{D^2} 
          = \left(\frac{L}{L_i D} - \frac{1}{D}\right) x_i
          = \frac{1}{D} \left( \frac{L}{L_i} - 1\right) x_i. \]
We have $L > L_1$ iff $\sum_{j \ge 2} x_j > 0$. Thus, the derivative of $x_1$ is
zero if $x_1 = 1$ and positive if $x_1 < 1$. Thus, $\sum_{i \ge 2} x_i$
decreases along all trajectories, starting in non-equilibrium points. \smallskip

Let $V = \sum_i x_i L_i$. Then,
\[ \dot{V} = \sum_i \frac{1}{D}\left(\frac{L}{L_i} - 1\right)x_i L_i = \frac{1}{D}
\sum_i (L - L_i) x_i. \]
So, it suffices to show $\sum_i L_i x_i \ge L = 1/\sum_i x_i/L_i$, or
equivalently,
$(\sum_i L_i x_i) (\sum_i x_i/L_i) \ge 1$. This is an immediate consequence of
the Cauchy-Schwarz inequality. Namely,
\[ 1 = \left(\sum_i \sqrt{x_i L_i} \sqrt{x_i/L_i}\right)^2  \le  \left(\sum_i
(\sqrt{x_i L_i})^2\right)\cdot \left(\sum_i ( \sqrt{x_i/L_i})^2 \right).\]\smallskip

Now, let $V = 1/L = \sum_j x_j/L_j$. We show that $V$ is increasing. We have 
\[ \dot{V} = \sum_i \frac{\dot{x_i}}{L_i} = \frac{1}{D} \sum_i \left(\frac{L}{L_i} -
1 \right)\frac{x_i}{L_i} = 
\frac{1}{D} \sum_i \left(\frac {L x_i}{L_i}\frac{1}{L_i} -
\frac{x_i}{L_i}\right) .\]
Let $z_i = Lx_i/L_i$. Then, $z_i \ge x_i$ if $L \ge L_i$, and $z_i \le x_i$ if $L
\le L_i$. Also $\sum_i z_i = 1$. Thus,
\[ D\cdot \dot{V} = \sum_i \frac{z_i - x_i}{L_i} = 
\sum_{i: L \ge L_i} \frac{z_i - x_i}{L_i} + \sum_{i: L < L_i} \frac{z_i -
x_i}{L_i} \ge \sum_{i: L \ge L_i} \frac{z_i - x_i}{L} + \sum_{i: L < L_i} \frac{z_i -
x_i}{L} = 0.\]
Moreover, $\dot{V} = 0$ if and only if $z_i = x_i$ for all $i$ if and only if
$x$ is a unit vector. \smallskip

Consider next the function $\sum_i Q_i L_i$. Then, 
\[ \sum_i Q_i L_i = \sum_i \Delta \frac{D_i}{L_i} L_i = \Delta D = \frac{D}{\sum_i
\frac{D_i}{L_i}} = \frac{1}{\sum_i \frac{x_i}{L_i}} = L; \]
hence, $\sum_i Q_i L_i$ is decreasing. \smallskip

The function $\Delta \sum_i D_i L_i = L \cdot \sum_i x_i L_i$ is the product of decreasing
functions and hence decreasing. \smallskip

Finally, let $V = \sum_{i \ge 2} (L_i \ln D_i - L_1 \ln D_1)$. Then
\begin{align*}
\dot{V} & = \sum_{i \ge 2} \left(L_i\frac{\dot{D_i}}{D_i} - L_1\frac{\dot{D_1}}{D_1}\right)
= \sum_{i \ge 2} \left(L_i\frac{Q_i - D_i}{D_i} - L_1\frac{Q_1 -
D_1}{D_1}\right) \\
& = 
\sum_{i \ge 2} \left(L_i\frac{D_i\Delta/L_i - D_i}{D_i} - L_1\frac{D_1\Delta/L_1 -
D_1}{D_1}\right) = \sum_{i \ge 2} (L_1 - L_i) < 0.\end{align*}
\Lnegskip\end{proof}

The Lyapunov function $\sum_{i \ge 2} (L_i \ln D_i - L_1 \ln D_1)$ was already
considered in~\cite{Miyaji-Ohnish07}. 

\begin{theorem}[Miyashi-Ohnishi~\cite{Miyaji-Ohnish07}] For a network of
parallel links, the dynamics converge against $D_1 = 1$ and $D_i = 0$ for $i
\ge 2$. \end{theorem}
\begin{proof} $x_1 = D_1/D$ is monotonically increasing and bounded by 1. Hence,
it converges. Assume that the limit $x_1^*$ is less than one. Clearly, $x_1^* > 0$. For
$x_1 \le x_1^*$, we have $1/L = \sum_i x_i/L_i \le x^*_1/L_1 + (1 -
x^*_1)/L_2$. Moreover, for large enough $t$, $x_1 \ge x_1^*/2$ and $D \le 2$
(Lemma \ref{lem:d-sum}),
and hence, $\dot{x_1} \ge \varepsilon$ for some $\varepsilon > 0$. Thus, $x_1^*
< 1$ is impossible. \end{proof}

Some of the Lyapunov functions have natural interpretations: $\sum_i Q_i
L_i$ is the total cost of the flow; $(\sum_i D_iL_i)/\sum_i D_i$ is the total
hardware cost normalized by the total diameter, where
a link of length $L$ and diameter $D$ has cost $DL$; and
$\Delta \sum_i D_iL_i$ is the potential difference between
source and sink multiplied by total hardware cost. These functions are readily
generalized to general networks by interpreting the summations as summations
over all edges of the network. Our computer simulations showed that none of
these functions is a Lyapunov function for general networks. 

However, $\sum_i D_i$ can also be interpreted as the minimum capacity of a
source-sink cut in a network where $D_i$ is the capacity of edge $i$. With this
interpretation, $(\sum_i D_iL_i)/\sum_i D_i$ becomes
\[              \frac{\sum_e D_e L_e}{\min_{S \in \cal C} C_S} ,\]
where $\cal C$ is the set of all $\source$-$\sink$ cuts and $C_S$ is the
capacity of the cut $C$. Our computer simulations suggested that this function
may serve as a Lyapunov function for general graphs. We will see below that a
slight modification is actually a Lyapunov function.

\section{Electrical Networks and Simple Facts}\label{sec: preliminaries}

In this section, we establish some more notation, review basic properties of
electrical networks, and prove some simple facts. 

Each node $v$ of the graph $G$ has a potential $p_v$ that is a function of
time. A potential difference $\Delta_e$ between the endpoints
of an edge $e$ induces a flow on the edge. For $e = (u,v)$,   
\begin{equation}
\label{eq:ohm}
Q_e = D_e \Delta_e / L_e = D_e (p_u - p_v)/L_e  = (p_u - p_v)/R_e 
\end{equation}
is the flow across $e$ in the direction from $u$ to $v$. If $Q_e <0$, the flow
is in the reverse direction. The potentials are such that there is flow
conservation in every vertex except for $\source$ and $\sink$ and such that the net flow from $\source$ to $\sink$ is one, that is, for every vertex $u$, we have  
\begin{equation}
\label{eq:kcl}
   \sum_{v:(u,v) \in E} Q_{u,v}  = \diverg(u),   
\end{equation}
where $\diverg(\source) = 1 = - \diverg(\sink)$ and $\diverg(u) = 0$ for all other vertices $u$. After fixing one potential to an arbitrary value, say $p_{\sink} = 0$, the other potentials are readily determined by solving a linear system. This means that each $Q_e$ can be expressed as a function of $R$ only.

For the main convergence proof, we will use some fundamental principles from
the theory of electrical networks (for a complete treatment, see for example \cite[Chapters II, IX]{Bollobas:1998}). \smallskip

\noindent
\textbf{Thomson's Principle.} The flow $Q$ is uniquely determined as a feasible flow that minimizes the total energy dissipation $\sum_e R_e Q^2_e$, with $R_e = L_e/D_e$. In other words, for any flow $x$ satisfying \eqref{eq:kcl}, 
\begin{equation}
\label{eq:thomson}
\sum_e R_e Q_e^2 \le \sum_e R_e x_e^2. 
\end{equation}\smallskip

\noindent
\textbf{Kirchhoff's Theorem.} 
For a graph $G=(N,E)$ and an oriented edge $e=(u,v) \in E$, let 
\begin{itemize}
\item $\sptree$ be the set of all spanning trees of $G$, and let
\item $\sptree(u,v)$ be the set of all spanning trees $T$ of $G$, for which the oriented edge $(u,v)$ lies on the unique path from $s_0$ to $s_1$ in $T$.  
\end{itemize} 
For a set of trees $S$, define $\Gamma(S) = \sum_{T \in S} \prod_{e \in T} D_e/L_e$. 
Then, the current through the edge $e$ is   
\begin{equation}
\label{eq:matrix-tree}
Q_{uv} = \frac{ \Gamma(\sptree(u,v))  - \Gamma(\sptree(v,u))}{ \Gamma(\sptree) }. 
\end{equation}\smallskip

\noindent
\textbf{Gronwall's Lemma. }
Let $\alpha, \beta \in \Real$ and let $x$ be a continuous differentiable real function on $[0,\infty)$. 
If $\alpha x(t) \le \dot{x}(t) \le \beta x(t)$ for all $t \ge 0$, 
then $$x(0) \, e^{\alpha t} \le x(t) \le x(0) \, e^{\beta t} \quad \text{ for all } t \ge 0.$$  
\begin{proof}
$$ \frac{d}{dt} \frac{x}{e^{\beta t}} = \frac{\dot{x} e^{\beta t} - \beta x e^{\beta t}}{e^{2 \beta t}} \le 0 \Rightarrow \frac{x(t)}{e^{\beta t}} \le \frac{x(0)}{e^{\beta 0}} = x(0). $$ 
A similar calculation establishes $x(t) \ge x(0) e^{\alpha t}$. 
\end{proof}

The next lemma gives some properties that are easily derived from
\eqref{dynamics}, \eqref{eq:ohm}, and \eqref{eq:kcl}. Recall that $\calC$ is
the set of \source-\sink\ cuts and $\capa(S)=\sum_{e \in \delta(S)} D_e$. Also,
let $\Lmin=\min_e L_e$, $\Lmax=\max_e L_e$, $n=|N|$, and $m=|E|$.  
\begin{lemma}
\label{lem:basic-properties}
The following hold for any edge $e \in E$ and any cut $S \in \calC$: 
\begin{enumerate}[(i)]
\item
\label{lem:flow-bound}
$|Q_e| \leq 1$. 
\item
\label{lem:source-flow}
$\sum_{e \in \delta(\{s_0\})} \abs{Q_e} = 1$. 
\item 
$D_e(t) \ge D_e(0) \, \exp(-t)$ for all $t$,  
\item
\label{lem:boundedness}
$D_e(t) \le 1 + (D_e(0) - 1) \exp(-t)$ for all $t$. 
\item
\label{lem:r-bounded}
$R_e \ge \Lmin / 2$ for all sufficiently large $t$. 
\item
\label{lem:cut}
$C_S(t) \ge 1 + (C_S(0)-1) \, \exp(-t)$ for all $t$, with equality if $S=\{s_0\}$.   
\item
\label{lem:source-cut}
$C_{\{s_0\}} \to 1$ as $t \to \infty.$ 
\item 
\label{lem:large diameter path}
Orient the edges according to the direction of the flow. 
For sufficiently large $t$, there is a directed source-sink path, in which all
edges have diameter at
least $1/2m$. 

\item
\label{lem:Delta-bound}
$\abs{\Delta_e} \le 2 nm \Lmax$ for all sufficiently large $t$.  

\item
\label{lem:ddot-over-d}
$\dot{D_e}/D_e \in [-1, 2nm\Lmax/\Lmin]$ for all sufficiently large $t$. 
\end{enumerate}
\end{lemma}

\begin{proof}
\begin{enumerate}[(i)]
\item 
Since $Q$ is a flow, it can be decomposed into $s_0$-$s_1$ flow paths and
cycles. If $|Q_e| > 1$, since $\diverg(s_0)=1$, there exists a positive cycle in
this decomposition, a contradiction to the existence of potential values at
the nodes. The claim is also an immediate consequence of \eqref{eq:matrix-tree}. 
\item
It follows from equations \eqref{eq:ohm} and \eqref{eq:kcl} that $p_{s_0}=\max_v p_v$, so $Q_{s_0,v} \ge 0$ for all $\{s_0,v\} \in E$, and $\sum_{e \in \delta(\{s_0\})} \abs{Q_e} = \sum_{e \in \delta(\{s_0\})} Q_e = 1$.  
\item 
From the evolution equation \eqref{dynamics}, $\dot{D}_e \ge - D_e$. The claim follows by Gronwall's Lemma. 
\item
$\abs{Q_e} \le 1$ for any edge $e$, so $\dot{D}_e \le 1 - D_e$ from \eqref{dynamics}, and the claim follows as before.  
\item
From \eqref{lem:boundedness}, $D_e \le 2$ for all sufficiently large $t$, so $R_e = L_e/D_e \ge \Lmin/2$ for the same $t$'s. 
\item
$\dot{C}_S = \sum_{e \in \delta(S)} \dot{D}_e = \sum_{e \in \delta(S)} (\abs{Q_e} - D_e) \ge 1 - C_S $, with equality if $S=\{s_0\}$.
\item
Follows by noting that the inequality in \eqref{lem:cut} becomes tight for the cut $\{s_0\}$, due to \eqref{lem:source-flow}. 
\item
From \eqref{lem:cut}, eventually $C_S \ge 1/2$ for all $S \in \calC$, so there
is an edge of diameter at least $1/2m$ in every cut. Thus, there is a
\source-\sink\ path in which every edge has diameter at least $1/2m$.
\item Consider a source-sink path in which every edge has diameter at least
$1/2m$. By \eqref{eq:ohm} the total potential drop $p_{\source}- p_{\sink}$ is at most $2nm \Lmax$. 
\item 
$ \dot{D_e}/D_e = (\abs{Q_e} - D_e) / D_e = \abs{\Delta_e}/L_e - 1$, and the bound follows from \eqref{lem:Delta-bound}. 
\end{enumerate}\negskip
\end{proof}

\section{Convergence}\label{sec:Convergence-General}

We will prove convergence for general graphs. Throughout this section, we 
will assume that $t$ is large enough 
for all the claims of Lemma \ref{lem:basic-properties} requiring a sufficiently large $t$
to hold.

\subsection{Properties of Equilibrium Points.}
Recall that $D \in \Real_+^E$ is an \emph{equilibrium point}, when $\dot{D_e}=0$ for all $e \in E$, which by \eqref{dynamics} is equivalent to $D_e = \abs{Q_e}$ for all $e \in E$. 

\begin{lemma}
\label{lem:equilibrium-mincut}
At an equilibrium point, $\min_{S \in \calC} \capa(S) = \capa(\{s_0\}) = 1$.
\end{lemma}
\begin{proof} $$1 \le \min_{S \in \calC} \sum_{e \in \delta(S)} \abs{Q_e} = \min_{S \in \calC} \capa(S) \le \capa(\{s_0\}) = \sum_{e \in \delta(\{s_0\})} \abs{Q_e} = 1.$$ 
\negskip\end{proof}

\begin{lemma}\label{lem:characterization of equilibria} 
The equilibria are precisely the flows of value 1, in which all
source-sink paths have the same length. If no two source-sink paths have the
same length, the
equilibria are precisely the simple source-sink paths. 
\end{lemma}
\begin{proof} Let $Q$ be a flow of value 1, in which all source-sink paths have the
same length. We orient the edges such that $Q_e \ge 0$ for all $e$ and show
that $D = Q$ is an equilibrium point. Let $E_1$
be the set of edges carrying positive flow, and let $V_1$ be the set of vertices
lying on a source-sink path consisting of edges in $E_1$. For $v \in V_1$, set
its potential to the length of the paths from $v$ to $\sink$ in $(V_1,E_1)$;
observe that all such paths have the same length by assumption. Let $Q'$ be the
electrical flow induced by the potentials and edge diameters. For any 
edge $e
= (u,v) \in E_1$, we have $Q'_e = D_e \Delta_e/L_e = D_e = Q_e$. Thus, $Q' = Q$. For any edge $e
\not\in E_1$, we have $Q_e = 0 = D_e$. We conclude that $D$ is an equilibrium point.

Let $D$ be an equilibrium point and let $Q_e$ be the corresponding current
along edge $e$, where we orient the edges so that $Q_e \ge 0$ for all $e \in
E$. Whenever $D_e>0$, we have $\Delta_e = Q_e L_e / D_e = L_e$ because of the
equilibrium condition. Since all directed $\source$-$\sink$ paths span the same
potential difference, all directed paths from $\source$ to $\sink$ in 
$\{e \in E : D_e > 0\}$ have the same length. Moreover, by Lemma
\ref{lem:equilibrium-mincut}, $\min_S \capa(S)=1$. Thus, $D$ is a flow of value 1.  
\end{proof}

Let $\Estar$ be the set of flows of value
one in the network of shortest source-sink paths. If the shortest source-sink
path is unique, $\Estar$ is a singleton, namely the flow of value one along the
shortest source-sink path. 

\subsection{The Convergence Process}

\begin{lemma}
\label{lem:monotonicity-sourcecut}
Let $W = (\capa(\{s_0\})-1)^2$. Then, $\dot{W} = -2 W \le 0$, with equality iff 
$\capa(\{s_0\})=1$. 
\end{lemma}
\begin{proof} 
Let $C_0 = \capa(\{s_0\})$ for short. Then, since $\sum_{e \in \delta(\{s_0\})} \abs{Q_e}=1$,  
$$ 
\dot{W} = 2(C_0 -1) \sum_{e \in \delta(\{s_0\})} \left(\abs{Q_e} - D_e
\right) = 2(C_0 -1) (1 - C_0) = -2 (C_0 -1)^2 \le 0. 
$$\Lnegskip
\end{proof}

\noindent 
The following functions play a crucial role. Let $C=\min_{S \in \calC} C_S$, and  
\begin{align*}
V_S &= \frac{1}{C_S} \sum_{e \in E} L_e D_e \text{ for each } S \in \calC, \\
V &= \max_{S \in \calC} V_S + W, \text{ and}\\
h &= - \frac{1}{C} \sum_{e \in E} R_e \abs{Q_e} D_e + \frac{1}{C^2} \sum_{e \in E} R_e D_e^2. 
\end{align*}


\begin{lemma}
\label{lem:equality}
Let $S$ be a minimum capacity cut at time $t$. Then, $\dot{V}_S(t) \le -h(t)$. 
\end{lemma}
\begin{proof}
Let $X$ be the characteristic vector of $\delta(S)$, that is, $X_e = 1$ if $e \in \delta(S)$ and 0 otherwise. Observe that $C_S=C$ since $S$ is a minimum capacity cut.  
We have
\begin{align*}
\dot{V}_S &= \sum_e \frac{\partial V_S}{\partial D_e} \dot{D}_e \\
&= \sum_e \frac{1}{C^2} \left( L_e C - \sum_{e'} L_{e'} D_{e'} X_e \right)
\left( \abs{Q_e} - D_e \right) \\
&= \frac{1}{C} \sum_e L_e \abs{Q_e} - \frac{1}{C^2} \left(\sum_{e'} L_{e'} D_{e'}
\right) \left(\sum_e X_e \abs{Q_e} \right) + \\
& \qquad\qquad
- \frac{1}{C} \sum_e L_e D_e  + \frac{1}{C^2} \left(\sum_{e'} L_{e'} D_{e'} \right)
\left(\sum_e X_e D_e \right) \\
&\le \frac{1}{C} \sum_e R_e \abs{Q_e} D_e  - \frac{1}{C^2} \sum_{e} R_{e} D_{e}^2
- \frac{1}{C} \sum_e L_e D_e  + \frac{1}{C} \sum_{e} L_{e} D_{e}  \\
&= -h. 
\end{align*}
The only inequality follows from $L_e=R_e D_e$ and $\sum_e X_e \abs{Q_e} \ge
1$, which holds because at least one unit current must cross $S$.   
\end{proof}

\begin{lemma}
\label{lem:max-rule}
Let $f(t) = \max_{S \in \calC} f_S(t)$, where each $f_S$ is continuous and  differentiable. If $\dot{f}(t)$ exists, then there is $S \in \calC$ such that $f(t)=f_S(t)$ and $\dot{f}(t)=\dot{f_S}(t)$. 
\end{lemma}
\begin{proof}
Since $\calC$ is finite, there is at least one $S \in \calC$ such that for each
fixed $\delta>0$, $f(t+\varepsilon)=f_S(t+\varepsilon)$ for infinitely many
$\varepsilon$ with $\abs{\varepsilon} \le \delta$. By continuity of $f$ and
$f_S$, this implies $f(t)=f_S(t)$. Moreover, since 
$$ \lim_{\varepsilon \to 0} \frac{\max_{S'} f_{S'}(t+\varepsilon) - \max_{S'} f_{S'}(t)}{\varepsilon} $$
exists and is equal to $\dot{f}(t)$, any sequence $\varepsilon_1,\varepsilon_2,\dots$ converging to zero has the property that 
$$\frac{\max_{S'} f_{S'}(t+\varepsilon_i) - \max_{S'} f_{S'}(t)}{\varepsilon_i} \to \dot{f}(t) \qquad \text{ for } i \to \infty. $$
Taking $(\varepsilon_i)_{i=1}^{\infty}$ to be a sequence converging to zero such that $f(t+\varepsilon_i)=f_S(t+\varepsilon_i)$ for all $i$, we obtain 
$$ \dot{f}(t) = \lim_{i \to \infty} \frac{f_S(t+\varepsilon_i) - f_S(t)}{\varepsilon_i} = \dot{f_S}(t). $$
\Lnegskip\end{proof}

\begin{lemma}
\label{lem:lyapunov}
$\dot{V}$ exists almost everywhere. 
If $\dot{V}(t)$ exists, then $\dot{V}(t) \le -h(t) -2W(t) \le 0$, and
$\dot{V}(t)=0$ if and only if $\dot{D}_e(t)=0$ for all $e$. 
\end{lemma}
\begin{proof}
$V$ is Lipschitz-continuous since it is the maximum of a finite set of
continuously differentiable functions. Since $V$ is Lipschitz-continuous, the
set of $t$'s where $\dot{V}(t)$ does not exist has zero Lebesgue measure (see
for example \cite[Ch.~3]{Clarke:1998}, \cite[Ch.~3]{Makela:1992}). When
$\dot{V}(t)$ exists, we have $\dot{V}(t) = \dot{W}(t)+\dot{V}_S(t)$ for some
$S$ of minimum capacity (Lemma \ref{lem:max-rule}). 
Then, $\dot{V}(t) \le -h(t) - 2W(t)$ by Lemmas~\ref{lem:monotonicity-sourcecut}
and \ref{lem:equality}.  

The fact that $W \ge 0$ is clear. We now show that $h \ge 0$. To this end, let
$F$ represent a maximum \source-\sink\ flow in an auxiliary network, having the
same structure as $G$, and where the capacity on edge $e$ is set equal to
$D_e$. In other words, $F$ is an \source-\sink\ flow satisfying $\abs{F_e} \le
D_e$ for all $e \in E$ and having maximum value. By the max-flow min-cut
theorem, this maximum value is equal to $C = \min_{S \in \calC} C_S$.   
But then, 
\begin{align*}
-h &= \frac{1}{C} \sum_e R_e \abs{Q_e} D_e  - \frac{1}{C^2} \sum_{e} R_{e} D_{e}^2 \\
&\le \frac{1}{C} \left(\sum_e R_e Q_e^2 \right)^{1/2} \left( \sum_e R_e D_e^2 \right)^{1/2} - \frac{1}{C^2} \sum_e R_e D_e^2 \\
&\le \frac{1}{C} \left( \sum_e R_e \frac{F_e^2}{C^2} \right)^{1/2} \left( \sum_e R_e D_e^2 \right)^{1/2} - \frac{1}{C^2} \sum_e R_e D_e^2 \\
&\le \frac{1}{C^2} \left( \sum_e R_e D_e^2 \right)^{1/2} \left( \sum_e R_e
D_e^2 \right)^{1/2} - \frac{1}{C^2} \sum_e R_e D_e^2 \\
&= 0, 
\end{align*}
where we used the following inequalities:
\begin{itemize}
\item[-] the Cauchy-Schwarz inequality $\sum_e (R_e^{1/2} \abs{Q_e}) (R_e^{1/2} D_e) \le (\sum_e R_e Q_e^2)^{1/2} (\sum_e R_e D_e^2)^{1/2}$;
\item[-] Thomson's Principle \eqref{eq:thomson} applied to the unit-value flows $Q$ and $F/C$; $Q$ is a minimum energy flow of unit value, while $F/C$ is a feasible flow of unit value; 
\item[-] the fact that $\abs{F_e} \le D_e$ for all $e \in E$.  
\end{itemize}

Finally, one can have $h=0$ if and only if all the above inequalities are
equalities, which implies that $\abs{Q_e}=\abs{F_e}/C=D_e/C$ for all $e$. And,
$W=0$ iff $\sum_{e \in \delta(\{s_0\})} D_e = 1 = \sum_{e \in \delta(\{s_0\})}
\abs{Q_e}$. So, $h=W=0$ iff $\abs{Q_e}=D_e$ for all $e$.   
\end{proof}


The next lemma is a necessary technicality. 
\begin{lemma}
\label{lem:h-uc}
The function $t \mapsto h(t)$ is Lipschitz-continuous.
\end{lemma}
\begin{proof}
Since $\dot{D_e}$ is continuous and bounded (by \eqref{dynamics}), $D_e$ is
Lipschitz-continuous. Thus, it is enough to show that $Q_e$ is
Lipschitz-continuous for all $e$.  

First, we claim that $D_e(t+\varepsilon) \le (1+2K \varepsilon) D_e$ for all
$\varepsilon \le 1/4K$, where $K=2nm\Lmax/\Lmin$.  
For if not, take $$\varepsilon=\inf \{ \delta \le 1/4K : D(t+\delta) > (1+2K\delta) D(t) \},$$
then $\varepsilon > 0$ (since $\dot{D_e}(t) \le K D_e(t)$ by Lemma
\ref{lem:basic-properties}) and, by continuity, $D_e(t+\varepsilon) \ge
(1+2K\varepsilon) D_e(t)$.  
There must be $t' \in [t,t+\varepsilon]$ such that $\dot{D_e}(t') = 2K D_e(t)$. On the other hand,
\begin{align*}
 \dot{D_e}(t') &\le K D_e(t') \le K (1+2K \varepsilon) D_e(t)   \\
 & \le K (1+2K/4K) D_e(t) < 2 K D_e(t), 
\end{align*} 
which is a contradiction. Thus, $D_e(t+\varepsilon) \le (1+2K \varepsilon) D_e$ for all $\varepsilon \le 1/4K$. Similarly, $D_e(t+\varepsilon) \ge (1-2K \varepsilon) D_e$. 

Consider now a spanning tree $T$ of $G$. Let $\gamma_T = \prod_{e \in T}
D_e/L_e$. Then $\gamma_T(t+\varepsilon) \le (1+2K \varepsilon)^n \gamma_T(t)
\le (1+4nK \varepsilon) \gamma_T(t)$ for sufficiently small
$\varepsilon$. Similarly, $\gamma_T(t+\varepsilon) \ge (1-4nK \varepsilon)
\gamma_T(t)$.  

By Kirchhoff's Theorem, 
$$ Q_{uv} = \frac{\sum_{T \in \sptree(u,v)} \gamma_T - \sum_{T \in \sptree(v,u)} \gamma_T}{\sum_{T \in \sptree} \gamma_T}, $$
and plugging the bounds for $\gamma_T(t+\varepsilon) / \gamma_T(t)$ shows that $Q_e(t+\varepsilon) = Q_e(t) (1+O(\varepsilon))$, where the constant implicit in the $O(\cdot)$ notation does not depend on $t$. Since $\abs{Q_e}\le 1$, we obtain that $\abs{Q_e(t+\varepsilon)-Q_e(t)} \le O(1) \cdot \varepsilon$, that is, $Q_e$ is Lipschitz-continuous, and this in turn implies the Lipschitz-continuity of $h$.   
\end{proof}

\begin{lemma}\label{lem: Q minus D} $\abs{D_e - \abs{Q_e}}$ converges to zero
for all $e \in E$. \end{lemma}\begin{proof}
Consider again the function $h$. We claim $h \to 0$ as $t \to \infty$. 
If not, there is $\varepsilon > 0$ and an infinite unbounded sequence $t_1,
t_2, \dots$ such that $h(t_i) \ge \varepsilon$ for all $i$. Since $h$ is
Lipschitz-continuous (Lemma \ref{lem:h-uc}), there is $\delta$ such that $h(t_i
+ \delta') \ge h(t_i) - \varepsilon/2 \ge \varepsilon/2$ for all $\delta' \in
[0,\delta]$ and all $i$. So by Lemma \ref{lem:lyapunov}, $\dot{V}(t) \le -h(t)
\le -\varepsilon/2$ for every $t$ in $[t_i,t_i+\delta]$ (except possibly a zero
measure set), meaning that $V$ decreases by at least $\varepsilon\delta/2$
infinitely many times. But this is impossible since $V$ is positive and
non-increasing.

Thus, for any $\varepsilon>0$, there is $t_0$ such that $h(t) \le \varepsilon$ for all $t \ge t_0$.  Then, recalling that $R_e \ge \Lmin/2$ for all sufficiently large $t$ (Lemma \ref{lem:basic-properties}.\ref{lem:r-bounded}), we find 
\begin{align*}
\sum_e \frac{\Lmin}{2} \left( \frac{D_e}{C} - \abs{Q_e} \right)^2 
&\le 
\sum_e R_e \left( \frac{D_e}{C} - \abs{Q_e} \right)^2 \\
&= \frac{1}{C^2} \sum_e R_e D_e^2 + \sum_e R_e Q_e^2 - \frac{2}{C} \sum_e R_e \abs{Q_e} D_e  \\
&\le \frac{2}{C^2} \sum_e R_e D_e^2 - \frac{2}{C} \sum_e R_e \abs{Q_e} D_e \\
&= 2 h \le 2 \varepsilon, 
\end{align*}
where we used once more the inequality $\sum_e R_e Q_e^2 \le \sum_e R_e D_e^2/C^2$, which was proved in Lemma \ref{lem:lyapunov}.  
This implies that for each $e$, $D_e/C - \abs{Q_e} \to 0$ as $t \to
\infty$. Summing across $e \in \delta(\{s_0\})$ and using Lemma
\ref{lem:basic-properties}.\ref{lem:source-flow}, we obtain $C_{\{s_0\}}/C - 1
\to 0$ as $t \to \infty$.  From Lemma \ref{lem:basic-properties}, $C_{\{s_0\}}
\to 1$ as $t \to \infty$, so $C \to 1$ as well. 

To conclude, we show that $D_e/C - \abs{Q_e} \to 0$ and $C \to 1$ together
imply $D_e - \abs{Q_e} \to 0$. Let $\varepsilon > 0$ be arbitrary. For all
sufficiently large $t$, $\abs{D_e/C - \abs{Q_e}} \le \varepsilon$, $\abs{1 - C} \le
\varepsilon$, $D_e \le 2$, and $C \ge 1/2$. Thus,
\[ \abs{D_e - \abs{Q_e}} \le \abs{D_e - D_e/C}+ \abs{D_e/C - \abs{Q_e}} \le D_e
\frac{\abs{C - 1}}{C} + \abs{D_e/C - \abs{Q_e}} \le 5 \varepsilon.\]
\Lnegskip\end{proof}

\begin{lemma}\label{lem:convergence of Delta} Let $\Delta = p_{\source} -
p_{\sink}$ be the potential difference between source and sink. $\Delta$
converges to the length $L^*$ of a shortest source-sink path.\end{lemma}
\begin{proof} Let $\cal L$ be the set of lengths of simple source-sink
paths. We first show that $\Delta$ converges to a point in $\cal L$ and then
show convergence to $L^*$.

Orient edges according to the direction of the flow. 
By Lemma~\ref{lem:basic-properties}.\ref{lem:large diameter path}, there is a directed 
source-sink path $P$ of edges of diameter at
least $1/2m$. Let $\varepsilon > 0$ be arbitrary. We will show $\abs{\Delta -
L_P} \le \varepsilon$. For this, it suffices to show $\abs{\Delta_e - L_e}
\le \varepsilon/n$ for any edge $e$ of $P$, where $\Delta_e$ is the potential
drop on $e$. By Ohm's law, the potential drop on $e$ is $\Delta_e = (Q_e/D_e)
L_e$, and hence, $\abs{\Delta_e - L_e} = \abs{Q_e/D_e - 1} L_e = \abs{(Q_e -
D_e)/D_e} L_e \le 2m \Lmax \abs{Q_e - D_e}$. The claim follows since $\abs{Q_e
- D_e}$ converges to zero.

The set $\cal L$ is finite. Let
$\varepsilon$ be positive and smaller than half the minimal distance between two elements in
$\cal L$. By the preceeding paragraph, there is for all sufficiently large $t$
a path $P_t$ such that 
$\abs{\Delta - L_{P_t}} \le \varepsilon$. Since $\Delta$ is a continuous
function of time, $L_{P_t}$ must become constant. We have now shown that
$\Delta$ converges to an element in $\cal L$. 

We will next show that $\Delta$ converges to $L^*$. 
Assume otherwise, and let $P'$ be a shortest undirected
source-sink path. Let $W_{P'} = \sum_{e \in P'} L_e \ln D_e$. This function was
already used by Miyaji and Ohnishi \cite{Miyaji-Ohnishi}. We have
\[ 
 \dot{W}_{P'} = \sum_{e \in P'} \frac{L_e}{D_e} (\abs{Q_e} - D_e)  = \sum_{e
\in P'} \abs{\Delta_e} - \sum_{e \in P'} L_e \ge p_{s_0}-p_{s_1} - L_{P'} =
\Delta - L^*.\]
Let $\delta > 0$ be such that there is no source-sink path with length in the
open interval $(L^*, L^* + 2 \delta)$. Then, $\Delta - L^*\ge \delta$
for all sufficently large $t$, and hence, $\dot{W}_{P'} \ge \delta$ for all
sufficiently large $t$. Thus, $W_{P'}$ goes to $+\infty$. However, $W_{P'} \le
n \Lmax$ for all sufficiently large $t$ since $D_e \le 2$ for all $e$ and $t$
large enough. This is a contradiction. Thus, $\Delta$ converges to $L^*$.
\end{proof}

\begin{lemma}\label{edges outside shortest paths die} 
Let $e$ be any edge that does not lie on a shortest source-sink
path. Then, $D_e$ and $Q_e$ converge to zero. \end{lemma}
\begin{proof} Since $\abs{D_e - \abs{Q_e}}$ converges to zero, it suffices to
prove that $Q_e$ converges to zero. Assume otherwise. Then, there is a
$\delta > 0$ such that $\abs{Q_e} \ge \delta$ for arbitrarily large
$t$. 

Consider any such $t$ and orient the edges according to the direction of the
flow at time $t$. Let $e = (u,v)$. Because of flow conservation, there must be
an edge into $u$ and an edge out of $v$ carrying flow at least
$Q_e/n$. Continuing in this way, we obtain a source-sink path $P$ in which
every edge carries flow at least $Q_e/n^n \ge \delta/n^n$; $P$ depends on
time and $L_P > L^*$ always. We will show $\abs{\Delta -
L_{P}} \le (L_P - L^*)/4$ for sufficiently large $t$, a contradiction to the fact that
$\Delta$ converges to $L^*$. 
For this, it suffices to show $\abs{\Delta_g - L_g}
\le (L_P - L^*)/(4n)$ for any edge $g$ of $P$, where $\Delta_g$ is the potential
drop on $g$. By Ohm's law, the potential drop on $g$ is $\Delta_g = (Q_g/D_g)
L_g$, and hence, $\abs{\Delta_g - L_g} = \abs{Q_g/D_g - 1} L_g = \abs{(Q_g -
D_g)/D_g} L_g \le \Lmax \abs{Q_g - D_g}/D_g$. For large enough $t$, $\abs{Q_g -
D_g} \le \min(\delta/(2n^n),\delta (L_P - L^*)/(8 n^{n+1}\Lmax))$. Then, 
$D_g \ge Q_g - \abs{Q_g - D_g} \ge
\delta/(2n^n)$, and hence, 
$\Lmax \abs{Q_g - D_g}/D_g \le (L_P - L^*)/(4n)$. \end{proof}

\begin{theorem}
\label{thm:convergence}\label{thm:main-apx}
The dynamics are attracted by $\Estar$. If the shortest source-sink path is
unique, the dynamics converge against a flow of value 1 on the shortest source
sink path. 
\end{theorem}
\begin{proof} $Q$ is a source-sink flow of value one at all times. We show
first that $Q$ is attracted to $\Estar$. Orient the edges in the
direction of the flow. We can decompose $Q$ into flowpaths. For an oriented
path $P$, let $1_P$ be the unit flow along $P$. We can write $Q = \sum_P x_p
1_P$, where $x_P$ is the flow along the path $P$. This decomposition is not
unique. We group the flowpath into two sets, the paths running inside $G_0$ and
the paths using an edge outside $G_0$, i.e., 
\[         Q = Q_0 + Q_1, \text{ where $Q_0 = \sum_{\text{$P$ is a path in
$G_0$}} x_P 1_P$}.\]
$Q_0$ is a flow in $G_0$, and each flowpath in $Q_1$ is a non-shortest
source-sink path.\footnote{The decomposition into $Q_0$ and $Q_1$ can be
constructed as follows: Initialize $Q_0$ to $Q$ and $Q_1$ to the empty flow. 
Consider any edge $e \not\in E_0$ carrying positive
flow in $Q_0$, say $\varepsilon$. Let $P$ be an oriented source-sink path carrying $\varepsilon$ units of flow and
using $e$. Add $\varepsilon 1_P$ to $Q_1$ and subtract it from $Q_0$. Continue until $Q_0$ is a
flow in $G_0$.} We show that the value of $Q_0$ converges to one. 

Assume otherwise. Then, there is a $\delta > 0$ such that the value of $Q_1$ is
at least $\delta$ for arbitrarily large times $t$. At any such time, there is an
edge $e \not\in E_0$ carrying flow at least $\delta/m$; this holds since
source-sink cuts contain at most $m$ edges. Since there are only finitely many
edges, there must be an edge $e \not\in E_0$ for which $Q_e$ does not converge
to zero, a contradiction to Lemma~\ref{edges outside shortest paths die}. 

We have now shown that the distance between $Q$ and $\Estar$ converges to
zero. By Lemma~\ref{lem: Q minus D},
$\abs{D_e - \abs{Q_e}}$ converges to zero for all $e$, and hence, the distance
between $Q$ and $D$ converges to zero. Thus, $D$ is attracted by $\Estar$. 

Finally, if the shortest source-sink path is unique, $\Estar$ is a singleton,
and hence, $D$ converges to the flow of value one along the shortest source-sink
path. 
\end{proof}

\begin{lemma}\label{lem: potentials in G0} If the shortest source-sink path is
unique, $p_v$ converges to $\dist(v,\sink)$ for each node $v$ on the shortest
source-sink path, where $\dist(v,\sink)$ is the shortest path distance from $v$ to $\sink$. \end{lemma}
\begin{proof} Let $P_0$ be the shortest source-sink path. For any $e \in P$,
$D_e$ converges to one and $\abs{D_e - Q_e}$ converges to zero. Thus, $\Delta_e$
converges to $L_e$. \end{proof}

\subsection{More on the Lyapunov Function $V$}
In this section, we study $V = \sum_e L_e D_e / C + (\capa(\{s_0\})-1)^2$ as a function of $D$. Recall that $C=C(D) = \min_{S \in \calC} \capa(S)$, where $\capa(S) = \sum_{e \in \delta(S)} D_e$. 

\begin{lemma}
Let $D^0$ and $D^1$ be two equilibrium points. Define 
$$ D^\lambda = (1-\lambda) D^0 + \lambda D^1, \qquad \lambda \in [0,1] .$$ 
If $V(D^{0}) < V(D^{1})$, then $V(D^\lambda)$ is a linear, increasing function of $\lambda$. 
\end{lemma}
\begin{proof}
By Lemma \ref{lem:equilibrium-mincut}, $C(D^0) = C(D^1) = 1$, and
$\mycapa_{D^0}(\{s_0\})=\mycapa_{D^1}(\{s_0\})=1$. Since $\mycapa_D(S)$ is
linear in $D$ for any fixed cut $S$, one has $\mycapa_{D^0}(S) \ge 1$ and
$\mycapa_{D^1}(S) \ge 1$, so $\mycapa_{D^\lambda}(S) \ge 1$ for all $S$. Thus,
$C(D^\lambda) \ge 1$.  
On the other hand, $\mycapa_{D^\lambda}(\{s_0\})=1$. Thus, $C(D^\lambda) = 1$,
and $V(D^\lambda) = \sum_e L_e D_e^\lambda$, that is, $V(D^\lambda)$ is a
linear function of $D^\lambda$.   
\end{proof}

\begin{lemma}\label{lem: minimizing V(D) is the same as shortest path}
The problem of minimizing $V(D)$ for $D \in \RR^E_+$ is equivalent to the shortest path problem. 
\end{lemma}
\begin{proof}
By introducing an additional variable $C = \min_{S} \capa(S) > 0$, the problem of minimizing $V(D)$ is equivalently formulated as
\begin{align*}
\min\, & \frac{1}{C} \sum_e L_e D_e + \left(\sum_{e \in \delta(\{\source\})} D_e - 1 \right)^2 \\
\text{s.t. } & \capa(S) \ge C \qquad \forall S \in \calC \\
& C>0 \\
& D \ge 0. 
\end{align*}
Substituting $x_e = D_e / C$, we obtain
\begin{align*}
\min\, & \sum_e L_e x_e  + C^{1/2} \left(\sum_{e \in \delta(\{\source\})} x_e - \frac{1}{C} \right)^2 \\
\text{s.t. } & \sum_{e \in \delta(S)} x_e \ge 1 \qquad \forall S \in \calC \\
& x \ge 0, C > 0,
\end{align*}
which is easily seen to be equivalent to the (fractional) shortest path problem. 
\end{proof}

Lemma~\ref{lem: minimizing V(D) is the same as shortest path} was the basis for
the generalization of our results to the transportation problem (Section \ref{sec:transportation}). We first
generalized the above Lemma to Lemma~\ref{lem: minimizing V(D) is the same as
transportation problem} and then used the Lyapunov function suggested by the
generalization.

\section{Rate of Convergence for Stable Flow Directions}\label{sec: stable flow
directions}

The direction of the flow across an edge depends on the initial conditions and
time. We do not know whether flow directions can change infinitely
often or whether they become ultimately fixed. 
In this section, we assume that flow directions stabilize and explore the
consequences of this assumption. We will be able to make quite precise
statements about the convergence of the system. We assume
uniqueness of the shortest source-sink path and add more non-degeneracy
assumptions as we go along.

An edge $e = \uedge{u,v}$ becomes \emph{horizontal} if $\lim_{t \to \infty} \abs{p_u - p_v} =
0$, and it becomes \emph{directed} from $u$ to $v$ (directed from $v$ to $u$)
if $p_u > p_v$ for all large $t$ ($p_v > p_u$ for all large $t$). An edge
\emph{stabilizes} if it either becomes horizontal or directed, and a network
\emph{stabilizes} if all its edges stabilize. If a network stabilizes, we
partition its edges into a set $E_h$ of horizonal edges and a set $\Eorient$ of
directed edges. If $\uedge{u,v}$ becomes directed from $u$ to $v$, then 
$(u,v) \in \Eorient$. 

We already know that the diameters of the edges on the
shortest source-sink path (we assume uniqueness in this section) converge to one. The diameters of the
edges outside $G_0$ converge to zero. The potential of a vertex $v \in G_0$
converges to $\dist(v,\sink)$. For stabilizing networks, we can prove a
lot more. In particular, we can predict the decay rates of edges, the limit
potentials of the vertices, and for each edge the
direction in which the flow will stabilize.

\begin{definition}[Decay Rate] Let $r \le 0$. 

A quantity $D(t)$ \emph{decays with rate at least}
$r$ if for every $\varepsilon > 0$ there is a constant $A > 0$ such that for
all $t$
\[         D(t) \le A e^{(r + \varepsilon) t}, \quad\text{or equivalently,}\quad \ln D(t) \le (\ln
A) + (r + \varepsilon) t.\]

A quantity $D(t)$ \emph{decays with rate at most} $r$ if for every
$\varepsilon > 0$ there is a constant $a > 0$
such that for all $t$
\[ D(t) \ge a e^{(r - \varepsilon) t}, \quad\text{or equivalently,}\quad \ln D(t) \ge  (\ln
a) + (r - \varepsilon) t.\]

A quantity $D(t)$ \emph{decays with rate} $r$ if it decays with rate at least and at
most $r$. 
\end{definition}

We first establish a simple Lemma that, for any edge, connects the decay rate
of the flow across the edge and the diameter of the edge. 

\begin{lemma}\label{decay of Q implies decay of D}\label{decay of Q implies
decay of D, a = 1}\label{similar decay of Q implies similar decay of D}
Let $-1 \le a < 0$ and let $e,g \in E$. If $Q_e$ decays with rate at least $a$, then so does $D_e$. $D_e$
decays with rate at most $-1$. If $\abs{\abs{Q_e} - \abs{Q_g}}$ decays with
rate at least $a$, then $\abs{D_e - D_g}$ decays
with rate at least $a$. \end{lemma}
\begin{proof} Assume first that $Q_e$ decays with rate at least $a$, where $-1
\le a < 0$. Then, for any $\varepsilon > 0$, there is an $A > 0$ such that $Q_e \le
A e^{(a + \varepsilon)t}$ for all $t$. Consider $f$ with $\dot{f} = A e^{(a +
\varepsilon)t} - f$. This has solution 
$f = f_0 e^{-t} + \alpha e^{(a + \varepsilon)t}$, where $\alpha = A/(1 + a + \varepsilon)$ and $f_0$
is determined by the value of $f$ at zero, namely, $f(0) = f_0 + \alpha$. 
Consider $D_e - f$. Then,
\[ \frac{d}{dt} ( D_e  - f) = \abs{Q_e} - D_e - (A e^{(a  + \varepsilon)t} - f) \le - (D_e - f) .\]
Thus, $D_e - f \le C' e^{-t}$ for some constant $C'$ by Gronwall's Lemma, and hence,
\[   D_e \le (f_0 + C') e^{-t} + \alpha e^{(a + \varepsilon)t} \le C'' e^{(a
+ \varepsilon)t} \]
for some constant $C''$. Thus, $D_e$ decays with rate at least $a$. \smallskip

$\dot{D_e} = \abs{Q_e} - D_e \ge -D_e$. Thus, $D_e$ decays with rate at most
$-1$ by Gronwall's Lemma. \smallskip

Finally, assume that $\abs{\abs{Q_e} - \abs{Q_f}}$ decays with rate at least
$a$. Then,
\[ \frac{d}{dt} (D_e - D_g) = \abs{Q_e} - \abs{Q_f} - (D_e - D_g) \le \abs{
\abs{Q_e} - \abs{Q_f}} - (D_e - D_g),\]
and therefore, $D_e - D_g$ decays with rate at least $-a$. The same argument
applies to $D_g - D_e$. 
\end{proof}

For a path $P$, let $W(P) \define \sum_{e \in P} L_e \ln D_e$ be its weighted
sum of log diameters, and let $\Delta(P) = p_a - p_b$ be the potential
difference between its endpoints. The function $W(P)$ was introduced by Miyaji
and Ohnishi~\cite{Miyaji-Ohnish07,Miyaji-Ohnishi}.

\ignore{\begin{lemma} $\dot{W}(P) = \sum_{e \in P} \abs{\Delta(e)} - L(P)$. If all
edges of $P$ are used in forward direction, $\dot{W}(P) = \Delta(P) - L(P)$. 
\end{lemma}
\begin{proof}  
\[ \dot{W}(P) = \sum_{e \in P} L_e \frac{\frac{D_e}{L_e} \abs{\Delta(e)} - D_e}{D_e}
= \sum_{e \in P} \left(\abs{\Delta(e)} - L_e\right).\] \end{proof}}

\begin{lemma}\label{lem: decay along undirected path} 
Let $P$ be an arbitrary path, let $\Delta(P)$ be the potential drop along $P$,
and let
$W(P) = \sum_{e \in P} L_e \ln D_e$. Then, 
\[   \dot{W}(P) = \Delta(P) - L(P) +  2
\sum_{e \in P:\ \Delta(e) < 0} \abs{\Delta(e)} .\]
If $\Delta(P) \le \Delta$ and $\Delta(e) \ge -\delta$ for some $\delta \ge 0$, all $e \in P$ and for all sufficiently large $t$, then 
\[   W(P)(t) \le C + (\Delta - L(P) + 2 n \delta) t  \]
for some constant $C$ and all $t$. 
If $\Delta(P) \ge \Delta$ for all sufficiently large $t$, then 
\[   W(P)(t) \ge C + (\Delta - L(P)) t  \]
for some constant $C$ and all $t$. 
\end{lemma}
\begin{proof} The first claim follows immediately from the dynamics of the
system. 
\[   \dot{W}(P) = \sum_{e \in P} \abs{\Delta(e)} - L(P) = \Delta(P) - L(P) +  2
\sum_{e \in P:\ \Delta(e) < 0} \abs{\Delta(e)}.\]

Let $t_0$ be such that $\Delta(P) \le \Delta$ and $\Delta(e) \ge -\delta$ for all $t
\ge t_0$. We integrate the equality from $t_0$ to $t$ and obtain
\[  W(P)(t) - W(P)(t_0) = \int_{t_0}^t  \dot{W}(P) dt \le (\Delta - L(P) +
2n \delta)(t - t_0).\]
This establishes the claim for $t \ge t_0$. Choosing $C$ sufficiently large
extends the claim to all $t$. 

Let $t_0$ be such that $\Delta(P) \ge \Delta$. We integrate the equality from $t_0$ to $t$ and obtain
\[  W(P)(t) - W(P)(t_0) = \int_{t_0}^t  \dot{W}(P) dt \ge (\Delta - L(P))(t - t_0).\]
This establishes the claim for $t \ge t_0$. Choosing $C$ sufficiently large
extends the claim to all $t$. 
\end{proof}

Edges that do not lie on a source-sink path never carry any flow, and
hence, their diameter evolves as $D_e(0) \exp(-t)$. From now on, we may
therefore assume that every edge of $G$ lies on a source-sink path. 

\begin{lemma}\label{lem: decay for Eh} For $e \in E_h$, $D_e$ decays with rate
$- 1$, and $\abs{Q_e}$ decays with rate at least $-1$. \end{lemma}
\begin{proof} We certainly have $D_e \le 2$ for all large $t$. Let $e = \uedge{u,v}$,
and let $\varepsilon > 0$ be arbitrary. 
Then, $\abs{p_u - p_v} \le \varepsilon L_e$ for all large $t$, and hence, $\abs{Q_e} = (D_e/L_e)\abs{p_u
- p_v} \le \varepsilon D_e$ for all large $t$. Thus, $\dot{D}_e \le (\varepsilon - 1)
D_e$ for all large $t$, and hence, $(d/dt) \ln D_e \le -1 + \varepsilon$. 
Thus, $D_e$ decays with rate at least $-1$. Since $\dot{D_e} \ge -D_e$, $D_e$
decays with rate at most $-1$. 

$\abs{Q_e} = (D_e/L_e) \abs{p_u - p_v} \le A D_e$ for some constant $A$. Thus, 
$\abs{Q_e}$ decays with rate at least $-1$. 
\end{proof}

We define a decomposition of $G$ into paths $P_0$ to $P_k$, an orientation of
these paths, a slope $f(P_i)$ for each $P_i$, a vertex
labelling $p^*$, and an edge labelling $r$. 
$P_0$ is a\footnote{We assume that $P_0$ is unique.}
shortest $\source$-$\sink$ path in $G$, $f(P_0) = 1$, $r_e = f(P_0)-1$ for all $e \in P_0$,
and $p_v^* = \dist(v,\sink)$ for all $v \in P_0$, where $\dist(v,\sink)$ is the
shortest path distance from $v$ to $\sink$. 
For $1 \le i \le k$, we have\footnote{We assume that $P_i$ is unique except if $f(P_i) = 0$.}
\[   P_i = \argmax_{P \in \cal P} f(P), \]
where $\cal P$ is the set of all paths $P$ in $G$ with the following properties:
\begin{itemize}
\item[-] the startpoint $a$ and the endpoint $b$ of $P$ lie on $P_0 \cup \ldots \cup
P_{i-1}$, $p^*_a \ge p_b^*$, and $f(P) = (p_a^* - p_b^*)/L(P)$;
\item[-] no interior vertex of $P$ lies on $P_0 \cup \ldots \cup
P_{i-1}$; and
\item[-] no edge of $P$ belongs to $P_0 \cup \ldots \cup
P_{i-1}$.
\end{itemize}
If $p^*_a > p_b^*$, we direct $P_i$ from $a$ to $b$. If $p^*_a = p^*_b$, we
leave the edges in $P_i$ undirected. We set 
$r_e = f(P_i) - 1$ for all edges of $P_i$, and $p^*_v = p^*_b +
f(P_i)\dist_{P_i}(v,b)$  for every interior vertex $v$ of
$P_i$. Here, $\dist_{P_i}(v,b)$ is the distance from $v$ to $b$ along path $P_i$. 
Figure~\ref{fig: path decomposition} illustrates the path 
decomposition. 

\begin{figure}[t]
\begin{center}
\psfrag{e1}{$e_1$}\psfrag{e2}{$e_2$}\psfrag{e3}{$e_3$}\psfrag{e4}{$e_4$}\psfrag{e5}{$e_5$}
\psfrag{e6}{$e_6$}\psfrag{s0}{$s_0$}\psfrag{s1}{$s_1$}\psfrag{u}{$u$}\psfrag{v}{$v$}\psfrag{w}{$w$}
\includegraphics[width=0.4\textwidth]{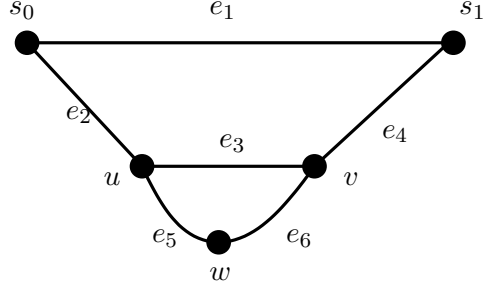}
\end{center}
\caption{\label{fig: path decomposition} 
All edges are assumed to have length
1; $P_0 = (e_1)$, $P_1 = (e_2,e_3,e_4)$, $P_2 = (e_5,e_6)$, $p^*_{s_0} = 1$,  $p^*_{s_1} = 0$, 
$p^*_v = 1/3$, $p^*_u = 2/3$, $p^*_w = 1/2$, $f(P_1) = 1/3$, and $f(P_2) =
1/6$. \protect\\
The path $(e_2,e_5,e_6,e_4)$ has $f$-value $1/4$.}
\end{figure}

\begin{lemma}\label{lem: path decomposition} There is an $i_0 \le k$ such that
\[   f(P_0) > f(P_1) >   \ldots >  f(P_{i_0}) > 0 = f(P_{i_0 + 1}) = \ldots =
f(P_k). \] \end{lemma}
\begin{proof} It suffices to show: if there is an $i$ such that 
$f(P_{i+1}) \ge f(P_i)$, then $f(P_i)= f(P_{i+1}) = 0$. 
If no endpoint of $P_{i+1}$ is an internal vertex of
$P_i$, then $f(P_{i+1}) = f(P_i)$; otherwise $P_{i+1}$ would have been chosen
instead of $P_i$. By assumption, equality is only possible if the $f$-values
are zero. So we may assume that at least one endpoint of $P_{i+1}$ is
an internal node of $P_i$; call it $c$ and assume w.l.o.g.~that it is the
startpoint of $P_{i+1}$. Split $P_i$ at $c$ into $P_i^1$ and $P_i^2$, and let
$d$ be the other endpoint of $P_{i+1}$; $d$ may lay on $P_i$. 

Assume first that $d$ does not lie on $P_i$ and consider the path $P_i^1
P_{i+1}$. The $f$-value of this path is
\[   \frac{p^*_a - p^*_d}{L(P_i^1) + L(P_{i+1})} = \frac{p^*_a - p^*_c + p^*_c -
p^*_d}{L(P_i^1) + L(P_{i+1})}. \]
Next, observe that $(p^*_a - p^*_c)/L(P_i^1) = f(P_i)$ since $p^*_c$ is defined by
linear interpolation and $(p^*_c - p^*_d)/L(P_{i+1}) = f(P_{i+1}) \ge f(P_i)$. 
In case of inequality, $P_i^1
P_{i+1}$ is chosen instead of $P_i$. In case of equality, there
are two paths with the same $f$-value. By assumption, this is only possible if the $f$-values
are zero. 

Assume next that $d$ also lies on $P_i$. We then split $P_i$ into three paths
$P_i^1$, $P_i^2$, and $P_i^3$ and consider the path $P_i^1 P_{i+1} P_i^3$. We
then argue as in the preceding paragraph. \end{proof}


\begin{theorem}\label{thm: convergence for stable networks} If a network stabilizes, then $\Eorient = \cup_{i \le i_0} E(P_i)$,
the orientation of any edge $e \in \Eorient$ agrees with the orientation 
induced by the path decomposition, and $E_h = \cup_{i > i_0} E(P_i)$. 
The potential of each node $v$ converges to $p^*_v$. The diameter of each edge $e
\in E \setminus P_0$ decays with rate $r_e$. \end{theorem}
\begin{proof} 
We use induction on $i$ to prove:
\begin{itemize}
\item[-] for every vertex $v \in P_0 \cup \ldots \cup P_i$, the node potential $p_v$ converges to
$p^*_v$;
\item[-] for every edge $e \in P_0 \cup \ldots \cup P_{\min(i,i_0)}$, the flow stabilizes in
the direction of the path $P_j$ containing $e$;
\item[-] for every edge $e \in P_1 \cup \ldots \cup P_i$, the diameter
converges to zero with rate $r_e$, and the flow converges to zero with rate at
least\footnote{If for an edge $e = \uedge{u,v}$, $p_u - p_v =0$ always, then
$Q_e = 0$ always. Thus, for horizontal edges, $Q_e$ may converge to zero faster
than with rate $-1$.} $r_e$. If $e \in P_i$ and $i \le i_0$, the flow converges to zero with
rate $r_e$. 
\end{itemize}

Lemma~\ref{lem: potentials in G0} establishes the base of the induction, the
case $i = 0$. Assume now that the induction hypothesis holds for $i - 1$; we
establish it for $i$. Let $P_{\le i-1} = P_0 \cup \ldots \cup P_{i-1}$. 

For $e \in E \setminus P_{\le i-1}$, let 
\[         f_e = \max \set{\frac{p^*_a - p^*_b}{L(P')}}{ P' \in {\cal P}_e}, \]
where ${\cal P}_e$ is the set of paths $P'$ in
$G \setminus P_{\le i-1}$ from some $a \in P_{\le i-1}$ to some $b \in P_{\le
i-1}$ with $p^*_a \ge p^*_b$ and containing $e$.
Then, $\max_{e \not\in P_{\le i-1}} f_e = f(P_i)$. For $i \le i_0$, we have
further $f(P_i) >  \max_{e \not\in P_{\le i}}
f_e \ge f(P_{i+1})$. In general, the last inequality may be strict; see
Figure~\ref{fig: path decomposition}.

\begin{lemma} For $e \in E \setminus P_{\le i-1}$, $\abs{Q_e}$ and $D_e$ decay with rate at least $f_e
- 1$. 
\end{lemma}
\begin{proof} According to Lemma~\ref{decay of Q implies decay of D}, it
suffices to prove the decay of $\abs{Q_e}$. Let $e \in
E \setminus P_{\le i - 1}$ and let $\varepsilon > 0$ be arbitrary. 
We need to show 
\[               \ln \abs{Q_e(t)} \le C + (f_e + \varepsilon -1)t \]
for some constant $C$ and all sufficiently large $t$. 

If $Q_e(t) = 0$, the inequality holds for any value of $C$. So assume $Q_e(t)
\not= 0$ and also assume that the flow across $e = \uedge{u,v}$ is in the direction from $u$ to
$v$. We construct a path $R(t)$ containing $uv$. 
For every vertex, except for source and
sink, we have flow conservation. Hence there is an edge $(v,w)$ carrying a flow
of at least $Q_e/n$ in the direction from $v$ to $w$. Similarly, there is an
edge $(x,v)$ carrying a flow of at least $Q_e/n$ in the direction from $x$ to
$v$. Continuing in this way, we reach vertices in $P_{\le i-1}$. Any
edge on the path $R(t)$ carries a flow of at least $Q_e/n^n$. 

Since potential differences are
bounded by $B:=2nm\Lmax$ (Lemma
\ref{lem:basic-properties}.\ref{lem:Delta-bound}), any edge $e'$ on $R(t)$ must
have a diameter of at least $Q_e
L_e/(n^n B) \ge (L_{\min}/(n^n B))Q_e$. Let $c = L_{\min}/(n^n B)$. Then, 
\[  W(R(t)) = \sum_{e' \in R(t)} L_{e'} \ln D_{e'} \ge L(R(t)) (\ln c + \ln
\abs{Q_e(t)}).\]
The path $R(t)$ depends on time. Let $a(t)$ and $b(t)$ be the endpoints of $R(t)$. 
Since $e$ does not belong to $P_{\le i-1}$,
\[   f(R(t)) = \frac{p^*_{a(t)} - p^*_{b(t)}}{L(R(t))} \le f_e. \]
For large enough $t$, we have $\Delta(R(t)) \le \Delta^*(R(t)) + \varepsilon
L(R)/2$. Every edge $e \in R(t)$ either belongs to $\Eorient$ or to $E_h$ 
due to the assumption that the network stabilizes. In the
former case, $R$ must use $e$ in the direction fixed in $\Eorient$, in the
latter case, the potential difference across $e$ converges to zero. We now invoke
Lemma~\ref{lem: decay along undirected path} with $\delta = \varepsilon
L(R)/(4n)$. It guarantees the existence of a constant $C_1$ such that 
\[   W(R(t))(t) \le C_1 + (\Delta^*(R(t)) + \varepsilon L(R)/2 - L(R) + \varepsilon L(R)/2) t  \]
for all $t$. The constant $C_1$ depends on the path $R(t)$. Since there are only
finitely many different paths $R(t)$, we may use the same constant $C_1$ for all
paths $R(t)$. 

Combining the estimates, we obtain, for all sufficiently large $t$, 
\[ L(R(t)) (\ln c + \ln
\abs{Q_e(t)})  \le  C_1 + (\Delta^*(R(t)) + \varepsilon L(R(t)) - L(R(t)) ) t,
\]
and hence,
\[ \ln \abs{Q_e(t)} \le C_1/L(R(t)) - \ln c  + (f_e + \varepsilon - 1) t.\]

\end{proof}

\begin{corollary}\label{cor:strict-rate}
 For $e \in E \setminus P_{\le i-1}$, $\abs{Q_e}$ and $D_e$
decay with rate at least $f(P_i) - 1$. 
If $i \le i_0$, then for any $e \in E \setminus P_{\le i}$, $\abs{Q_e}$ and $D_e$
decay with rate at least $f(P_i) - \delta - 1$ for some $\delta > 0$.
\end{corollary}
\begin{proof} If $i \le i_0$, and hence, $f(P_i) > 0$, $f_e < f(P_i)$ for any
edge $e \in E \setminus P_{\le i}$. The claim follows. \end{proof}

\begin{lemma}\label{lem:exact-rate}
 Let $e \in P_i$. Then, $D_e$ decays with
rate $f(P_i) - 1$. If $i \le i_0$, then $\abs{Q_e}$ decays with rate $f(P_i) -1$. \end{lemma}
\begin{proof} We distinguish the cases $f(P_i) = 0$ and $f(P_i) > 0$. 
If $f(P_i) = 0$, the diameter of all edges $e \in P_i$ decays with rate at
least $-1$ (Lemma \ref{lem: decay along undirected path}). No diameter decays
with a rate faster than $-1$. 

We turn to the case $f \assign f(P_i) > 0$. The flows across the edges 
in $E \setminus P_{< i}$ decay with 
rate at least $f - 1$, and the flows across the edges 
edges in $E \setminus P_{\le i}$ decay faster, say
with rate at least $f - \delta - 1$ for some positive $\delta$ (Corollary \ref{cor:strict-rate}). We first show
\begin{equation} \label{upperboundXXX}
                  W(P_i) \le C + L(P_i) \cdot \max(\ln D_e, (f - \delta - 1)t) 
\end{equation}
for sufficiently large $t$ and some constant $C$. 
If $P_i$ consists of a single edge $e$, $W(P_i) = L_e \ln D_e(t)$ and
(\ref{upperboundXXX}) holds. Assume next that $P_i = e_1 \ldots e_k$ with $k > 1$. 
Consider any interior node $u$ of the path. The flow into $u$ is equal to the
flow out of $u$, and $u$ has two incident edges\footnote{Here, we need
uniqueness of $P_i$. Otherwise we would have a network of paths with the same slope.}
in $P_i$. The flow on
the other edges incident to $u$ decays with rate at least $f - \delta - 1$. Thus
for any two consecutive edges on $P_i$, $\abs{\, \abs{Q_{e_j}} -
\abs{Q_{e_{j+1}}}\, }$ decays with rate at least $f - \delta - 1$. By
Lemma~\ref{similar decay of Q implies similar decay of D}, this implies that 
$\abs{D_{e_j} - D_{e_{j+1}}}$ decays with rate at least $f - \delta - 1$. 
Thus, we have 
$D_{e_j} = D_{e} + g_{e_j}$, where $\abs{g_{e_j}} \le C_1 e^{(f - \delta - 1)t}$
for some constant $C_1$ and all $j$. Plugging into the definition of
$W(P_i)$ yields
\begin{align*}
 W(P_i) & \le \sum_{e_j \in P_i} L_{e_j} \ln \left(2 \max(D_{e},g_{e_j})\right) \\
        & \le L(P_i) \ln 2 + L(P_i) \max( \ln D_{e}, \ln C_1 e^{(f - \delta - 1)t}),
\end{align*}
and we have established (\ref{upperboundXXX}). 

Let $t_0$ be large enough such that $\abs{\Delta(P_i) - \Delta^*(P_i)} \le
\delta L(P_i)/2$ for all $t \ge t_0$. Then, by Lemma~\ref{lem: decay along
undirected path}, 
\begin{equation} W(P_i) \ge A  + L(P_i) (f  - \delta/2 - 1) t
\label{eq: lower bound} \end{equation}
for some constant $A$ and all $t$. 

Combining (\ref{upperboundXXX}) and (\ref{eq: lower bound}) yields
\[  A  + L(P_i) (f  - \delta/2  - 1) t \le C + L(P_i) \cdot
\max(\ln D_e, (f - \delta - 1)t).\]
Thus, for every $t$ we have either
\[  A  + L(P_i) (f  - \delta/2 - 1) t \le C + L(P_i) \cdot \ln
D_e \]
or 
\[ A  + L(P_i) (f  - \delta/2 - 1) t \le C + L(P_i) \cdot (f -
\delta - 1)t .\]
The latter inequality does not hold for any sufficiently large $t$. Thus, the
former inequality holds for all sufficiently large $t$, and hence, $D_e$ decays
with rate at most $f(P_i) - 1$. By Lemma~\ref{decay of Q implies decay of D},
$\abs{Q_e}$ cannot decay at a faster rate if $f(P_i) > 0$. 
\end{proof}

\begin{lemma} For $v \in P_i$, the potentials converge to $p^*_v$. For $e \in
P_i$ and $i \le i_0$, the flow direction stabilizes in the direction of $P_i$. \end{lemma}
\begin{proof} Assume $i \le i_0$ first. Let $P_i = e_1 \ldots e_k$. The flows
and the diameters of the edges in $P_i$ decay with rate $f(P_i) - 1$ (Lemma \ref{lem:exact-rate}). The flows
and diameters of the edges incident to the interior vertices of $P_i$ and not
on $P_i$ decay faster, say with rate at least $f(P_i) - \delta - 1$, where
$\delta > 0$. For large $t$ and any interior vertex of $P_i$, one edge of $P_i$
must, therefore, carry flow into the vertex, and the other edge incident to the
vertex must carry it out of the vertex. Thus, the edges in $P_i$ must either
all be
directed in the direction of $P_i$ or in the opposite direction. As current
flows from higher to lower potential, they must be directed in the direction of
$P_i$. 

Because the flow and the diameters of the edges not on $P_i$ and incident to
interior vertices decay faster, we have for any $\varepsilon > 0$ and sufficiently large $t$
\[   Q_{e_j} = Q_{e_1}(1 + \varepsilon_j)\quad\text{and}\quad D_{e_j}=
D_{e_1}(1 + \varepsilon'_j),\]
where $\abs{\varepsilon_j},\abs{\varepsilon'_j} \le \varepsilon$. The 
potential drop $\Delta_{e_j}$ on edge $e_j$ is equal to 
\[  \Delta_{e_j} = \frac{Q_{e_j} L_{e_j}}{D_{e_j}} =  \frac{Q_{e_1}(1 +
\varepsilon'_j)}{D_{e_1}(1 + \varepsilon_j)} L_{e_j}, \]
and hence, the potential drop along the path is 
\[   p_a - p_b = \sum_j \Delta_{e_j} = \frac{Q_{e_1}}{D_{e_1}} L(P_i) (1 +
\varepsilon''), \]
where $\varepsilon''$ goes to zero with $\varepsilon$. 
The potential drop along the path converges to $p^*_a - p^*_b$. Thus,
${Q_{e_1}}/{D_{e_1}}$ converges to $f(P_i)$, and therefore, the potential of any
interior vertex $v$ of $P_i$ converges to $p^*_v$. 

We turn to the case $i > i_0$. The potentials of the endpoints of $P_i$
converge to the same value. Thus, the potentials of all interior vertices of
$P_i$ converge to the common potential of the endpoints. 
\end{proof}

We have now completed the induction step. \end{proof}

\section{The Wheatstone Graph}\label{sec: Wheatstone}

Do edge directions stabilize? We do not know. We know one graph class for which
edge directions are unique, namely series-parallel graphs. The simplest graph
which is not series-parallel is the Wheatstone graph shown in 
Figure~\ref{fig:wheatstone}. We use the following notation: 
We have edges $a$ to $e$ as shown in the figure. For an edge $x$, 
$R_x = L_x/D_x$ denotes its resistance and $C_x = D_x/L_x$ denotes its
conductance.\footnote{Observe that we use the letter $C$ with a different meaning
than in preceding sections.} For edges $a$, $b$, $c$, and $d$, the direction of the flow is
always downwards. For the edge $e$, the direction of the flow depends on the
conductances. We have an example where the direction of the flow across $e$
changes twice. 

A shortest path from source to sink may have two essentially different shapes. It
either uses $e$, or it does not. If $e$ lies on a shortest path, Lemma~\ref{lem: decay along
undirected path} suffices to prove convergence as observed
by~\cite{Miyaji-Ohnishi}. If $(a,e,d)$ is a shortest path\footnote{For
simplicity, we assume uniqueness of the shortest path in this section.},
let $P = (a,e)$ and $P' = (b)$. Then, 
\[ \frac{d}{dt} (W(P) - W(P')) \ge \Delta(P) - L(P) - (\Delta(P') - L(P')) = L(P')
- L(P) > 0.\]
Since $W(P)$ is bounded, this implies $W(P') \rightarrow -\infty$. Thus, $D_b$
converges to zero. Similarly, $D_d$ must converge to zero. More precisely,
$W(P')$ goes to $-\infty$ linearly, and hence, $D_b$ and similarly $D_d$ decay
exponentially. 

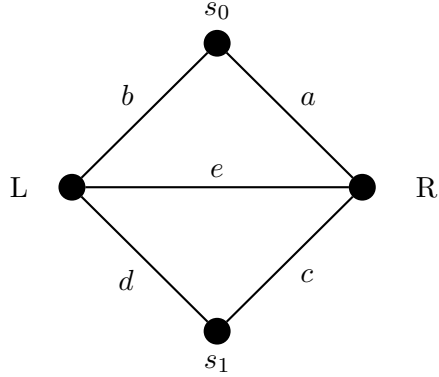
\begin{figure}[t]
\begin{center}
\begin{tikzpicture}[-,thick,node distance=3 cm] 
\tikzstyle{reg}=[circle,draw=black,fill=black,scale=0.9] 
\node (L)  [reg,label={left:L\ \ \ }] {}; 
\node (s0)  [reg,above right of=L,label={above:$\source$}] {}; 
\node (s1)  [reg,below right of=L,label={below:$\sink$}] {};
\node (R)  [reg,above right of=s1,label={right:\ \ \ R}] {};
\path (L)  edge node [label={above left:$b$}] {} (s0);
\path (L)  edge node [label={below left:$d$}] {} (s1); 
\path (s0)  edge node [label={above right:$a$}] {} (R); 
\path (s1)  edge node [label={below right:$c$}] {} (R); 
\path (L)  edge node [label={above:$e$}] {} (R);
\end{tikzpicture}
\end{center}
\caption{\label{fig:wheatstone} The Wheatstone graph.}
\end{figure}

The non-trivial case is that the shortest path does not use $e$. We may assume
w.l.o.g.~that the shortest path uses the edges $a$ and $c$. The ratio
\[   x_a = \frac{R_a}{R_a + R_c} = \frac{1}{1 + R_c/R_a} =
\frac{1}{1 + C_a/C_c} = \frac{C_c}{C_a + C_c}\]
is the ratio of the resistance of $a$ to the total resistance of the right
path; define $x_b$, $x_c$, and
$x_d$ analogously. Observe $x_a + x_c = 1$ and $x_b + x_d = 1$. 
Let 
\[    x_a^* = \frac{L_a}{L_a + L_c}; \]
define $x_b^*$, $x_c^*$, and $x_d^*$ analogously. Without edge $e$, the
potential drop on the edge $a$ is $x_a$ times the potential difference between
source and sink. If $D_a = D_c$, which we expect in the limit, $x_a =
x_a^*$.

\begin{lemma} Let $S = C_aC_b(C_c + C_d) + (C_a+ C_b)C_cC_d + (C_a + C_b)(C_c +
C_d)C_e$. Then, 
\begin{align*}
\dot{x_a} &=  \frac{C_a C_c }{S L_a L_c (C_a + C_c)^2} \left( (C_b + C_d +C_ e)(L_a +
L_c)(C_a + C_c)(x^*_a - x_a) + 
  C_e C_b L_c  \left( \frac{x_a^*}{x_c^*} - \frac{x_b}{x_d}
\right)\right)\\
\dot{x_b} &=  \frac{C_b C_d }{S L_b L_d (C_b + C_d)^2} \left( (C_a + C_c + C_e)(L_b +
L_d)(C_b + C_d)(x^*_b - x_b) + 
  C_e C_a L_d  \left( \frac{x_b^*}{x_d^*} - \frac{x_a}{x_c}
\right)\right).  \end{align*}     
\end{lemma}
\begin{proof} The derivatives of $C_a$ to $C_e$ were computed by 
Miyaji and Ohnishi~\cite{Miyaji-Ohnish07}: 
\begin{align*}
\dot{C_a} &= \frac{C_a}{S L_a} (C_b C_c + C_c C_d + C_c C_e + C_d C_e) - C_a \\
\dot{C_c} &= \frac{C_c}{S L_c} (C_a C_d + C_a C_b + C_a C_e + C_b C_e) - C_c.  
\end{align*}
The derivatives of $C_b$ and $C_d$ can be obtained from the above by symmetry (exchange $a$ with $b$ and $c$ with $d$). 
We now compute $\dot{x_a}$:  
\begin{align*}
\frac{d}{dt} \frac{C_c}{C_a + C_c} 
&= \frac{ - (\dot{C_a} C_c - C_a \dot{C_c})}{(C_a + C_c)^2}\\
&= \frac{{-\left(\frac{C_a }{S L_a}(C_bC_c + C_cC_d + C_cC_e + C_dC_e) -
C_a\right)C_c}}{(C_a+C_c)^2} + 
\\
& \qquad\qquad + \frac{C_a \left(\frac{C_c }{S L_c}(C_aC_d + C_aC_b + C_aC_e +
C_bC_e) - C_c\right)}{(C_a + C_c)^2}\\
&=\frac{C_a C_c }{S (C_a + C_c)^2} \left(\frac{C_aC_d + C_aC_b + C_aC_e + C_bC_e}{L_c} - \frac{C_bC_c +
C_cC_d + C_cC_e + C_dC_e}{L_a}\right)\\
&= \frac{C_a C_c }{S (C_a + C_c)^2} \left( (C_b + C_d + C_e)\left(\frac{C_a}{L_c} - \frac{C_c}{L_a}\right) +
C_e\left(\frac{ C_b}{L_c} - \frac{C_d}{L_a}\right) \right)\\
&= \frac{C_a C_c }{S L_a L_c (C_a + C_c)^2} \left( (C_b + C_d + C_e)(D_a - D_c) +
C_e ( C_b L_a - C_d L_c ) \right)\\
&= \frac{C_a C_c }{S L_a L_c (C_a + C_c)^2} \left( (C_b + C_d + C_e)(D_a - D_c) +
  C_e C_b L_c  \left( \frac{L_a}{L_c} - \frac{L_b/D_b}{L_d/D_d} \right)\right). 
\end{align*}
Finally, observe
\[ x^*_a - x_a = \frac{L_a}{L_a + L_c} - \frac{C_c}{C_a + C_c}= \frac{L_a(C_a + C_c) - C_c(L_a +
L_c)}{(L_a + L_c)(C_a + C_c)} = \frac{D_a - D_c}{(L_a + L_c)(C_a + C_c)} .\]\end{proof}

We draw the following conclusions:
\begin{itemize}
\item[-] if $C_e = 0$, then $\sign(\dot{x_a}) = \sign(D_a - D_c) = \sign(x_a^*
- x_a)$. Thus, $x_a$ converges monotonically against $x_a^*$. 
\item[-] From $x_b + x_d = 1$ and $x_a^* + x_c^* = 1$, we conclude
\[ \sign\left(\frac{x_a^*}{x_c^*} - \frac{x_b}{x_d}\right) =
\sign(x_a^* - x_b). \]
\item[-] if $s = \sign(x_a^* - x_b) = \sign(x_a^* - x_a)$, then
$\sign(\dot{x_a}) = s$. 
\item[-] if $x_a, x_b > x_a^*$, then $x_a$ decreases.
\item[-] if $x_a, x_b < x_a^*$, then $x_a$ increases.
\item[-] if $x_d, x_c > x_d^*$, then $x_d$ decreases (equivalent to: if $x_a, x_b < x_b^*$, then $x_b$ increases).
\item[-] if $x_d, x_c < x_d^*$, then $x_d$ increases (equivalent to: if $x_a, x_b > x_b^*$, then $x_b$ decreases).
\end{itemize}

\begin{theorem} Assume $x_a^* < x_b^*$, that is, $L_a/L_c <
L_b/L_d$. Then,
\begin{enumerate}
\item The regime $x_a, x_b > x_b^*$ cannot be entered. By symmetry, the
regime $x_a, x_b < x_a^*$ cannot be entered. 
\item In the regime $x_a, x_b \in [x_a^*,x_b^*]$, $x_a$ decreases and
$x_b$ increases. Hence, in this regime, the direction of the middle edge $e$ can change at most
once. 
\item If the dynamics stay in the regime $x_a, x_b \ge x_b^*$ forever, $x_a$ and
$x_b$ converge. 
\item If the dynamics stay in the regime $x_a, x_b \le x_a^*$ forever, $x_a$ and $x_b$
converge.
\end{enumerate}
\end{theorem}
\begin{proof} 
At (1): In the regime $x_a, x_b > x_b^*$, $x_a$ and $x_b$ both decrease, and
hence, the dynamics cannot enter the regime from the outside. More precisely, we consider two cases: 
$x_b \ge x_b^*$ and $x_a = x_b^*$, or $x_a > x_b^*$ and $x_b = x_b^*$. 

If $x_b \ge x_b^*$ and $x_a = x_b^*$, $x_a$ is non-increasing, 
and hence, we cannot enter the regime. 

If $x_a > x_b^*$ and $x_b = x_b^*$, $x_b$ is non-increasing, and hence, we cannot enter the regime. 

At (2): Obvious from the equations. 

At (3): Then, $x_a$ and $x_b$ are monotonically decreasing and hence
converging. The derivative of $x_b$ clearly goes to zero if $x_b$ and $x_a$
converge to $x_b^*$. 

At (4): Symmetrically to (3).
\end{proof} 

\begin{figure}[t]
\begin{center}
\psfrag{x_a}{$x_a$}\psfrag{x_b}{$x_b$}
\includegraphics[width=0.4\textwidth]{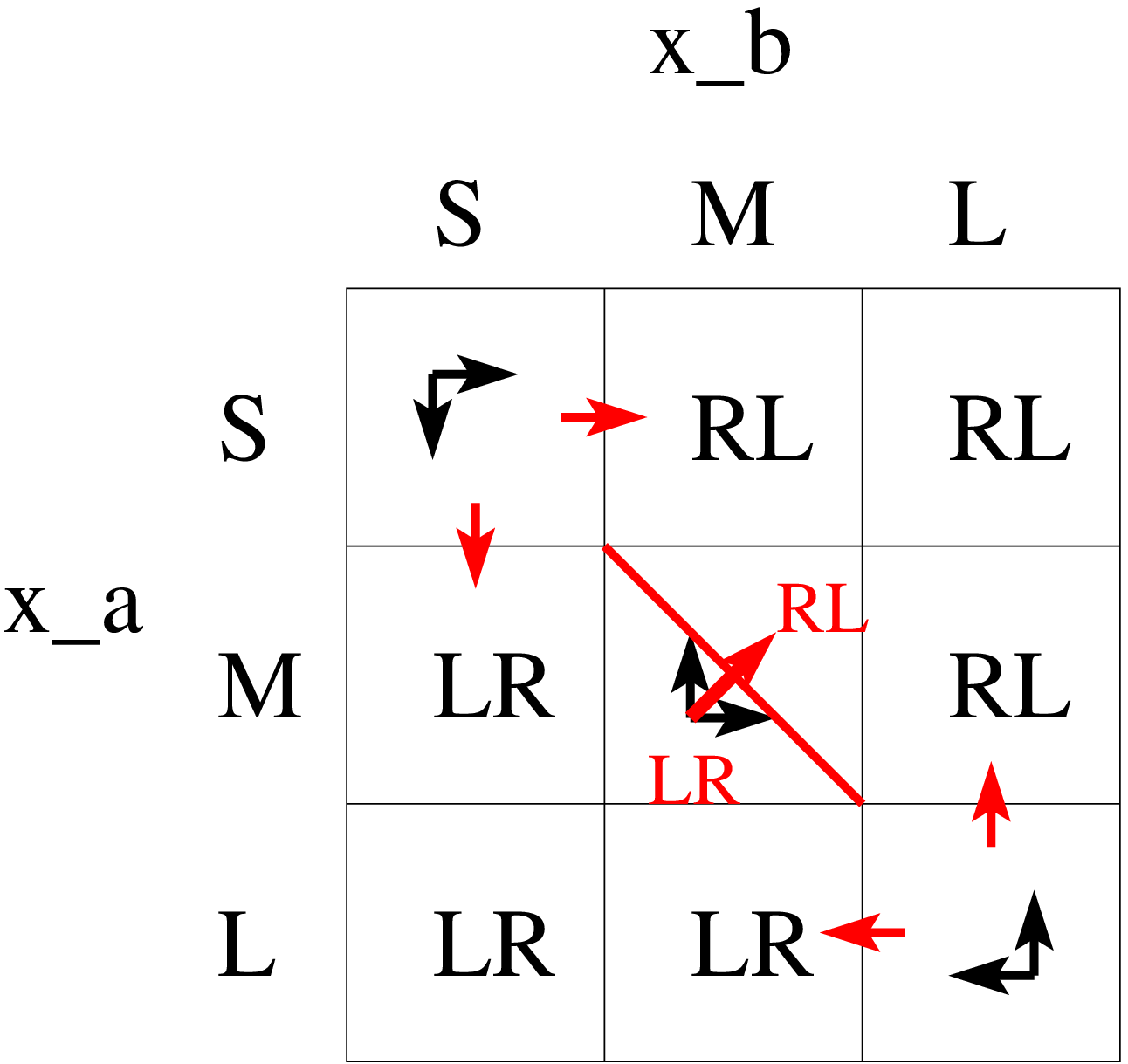}
\end{center}
\caption{\label{fig: Wheatstonetbl} The transition diagram under the assumption
$x_a^* < x_b^*$.}
\end{figure}

In Figure~\ref{fig: Wheatstonetbl}, we use $S$, $M$, and $L$ to denote the three
ranges: $S = [0,x_a^*]$, $M = [x_a^*,x_b^*]$, and $L = [x_b^*, 1]$. The box $M
\times M$ is divided into the triangles $x_a < x_b$ and $x_a > x_b$. The figure also
shows that the boxes $S \times S$ and $L \times L$ cannot be entered and that
the latter triangle cannot be entered from the former. 

We conclude the following dynamics: Either the process stays in $S \times S$ or
$L \times L$ forever or it does not do so. If it leaves these sets of states,
it cannot return. 
Moreover, there is no transition from the set of states RL to the set of states LR. Thus, if the process does not
stay in $S \times S$ or $L \times L$ forever, the direction of the middle edge
stabilizes. 

Assume now that the dynamics stay forever in $S \times S$, or in $L \times L$. Then, $x_a$ and $x_b$
converge. Let $x_a^\infty$ and $x_b^\infty$ be the limit values. If the limit
values are distinct, the direction of the middle edge stabilizes. If the limit
values are the same, the edge is horizontal and hence stabilizes. We summarize
the discussion. 

\begin{theorem}\label{thm: Wheatstone} The dynamics of the Wheatstone graph stabilize.\end{theorem}

\section{The Uncapacitated Transportation Problem} 
\label{sec:transportation}

The uncapacitated transportation problem generalizes the shortest path
problem. With each vertex $v$, a
supply/demand $b_v$ is associated. It is assumed that $\sum_v b_v = 0$. Nodes
with positive $b_v$ are called supply nodes and nodes with negative $b_v$ are
called demand nodes. In the shortest path problem, exactly two vertices have
non-zero supply/demand. A feasible solution to the transportation problem is a
flow $F$ satisfying the mass balance constraints, i.e., for every vertex $v$,
$b_v$ is equal to the net flow out of $v$. The cost of a solution is $\sum_e F_e
L_e$. The Physarum solver for the transportation problem is as follows: At any
fixed time, the current $Q$ is a feasible solution to the
transportation problem satisfying Ohm's law \eqref{eq:ohm}. The
dynamics evolve according to (\ref{dynamics}). 

For technical reasons, we extend $G$ by a vertex $\source$ with $b_{\source} =
1$, connect $\source$ to an arbitrary vertex $v$, and decrease $b_v$ by
one. The flow on the edge $(\source, v)$ is equal to one at all times. 

Our convergence proof for the shortest path problem extends to the
transportation problem. 
A cut $S$ is a set of vertices. The edge set
$\delta(S)$ of the cut is the set of edges having exactly one endpoint in $S$,
and the capacity $\capa(S)$ of the cut is the sum of the $D$-values in the
cut. The demand/supply of the cut is $b_S = \sum_{v \in S} b_v$. A cut $S$ is 
non-trivial if $b_S \not= 0$. We use $\calC$ to denote the family of
non-trivial cuts. For a non-trivial cut $S$, let $\CB_S = C_S/b_S$, and 
let $\CB = \min \set{\CB_S}{S \in \calC}$. One may
view $\CB$ as a scale factor; our transportation problem has a solution in a
network with edge capacities $D_e/\CB$. 
A cut $S$ with $\CB_S  = \CB$ is called a \emph{most constraining
cut}. 

\paragraph{Properties of Equilibrium Points.}
Recall that $D \in \Real_+^E$ is an \emph{equilibrium point} when $\dot{D_e}=0$ for all $e \in E$, which is equivalent to $D_e = \abs{Q_e}$ for all $e \in E$.  

\begin{lemma}
\label{lem:equilibrium-mincut-trasport}
At an equilibrium point, $\min_{S \in \calC} \capa(S)/\abs{b_S} = \capa(\{s_0\})/b_{\{s_0\}} = 1$.
\end{lemma}
\begin{proof} \[ 1 \le \min_{S \in \calC} \sum_{e \in \delta(S)} \frac{\abs{Q_e}}{\abs{b_S}} =
\min_{S \in \calC} \frac{\capa(S)}{\abs{b_S}} \le \frac{\capa(\{s_0\})}{b_{\{s_0\}}} = 1. \]
\Lnegskip\end{proof}

\begin{lemma}\label{lem:characterization of equilibria transportation} 
The equilibria are precisely the solutions to the transportation problem with
the following equal-length property: Orient the edges such that $Q_e \ge
0$ for all $e$, and let $N$ be the subnetwork of edges carrying positive
flow. Then, for any two vertices $u$ and $v$, all directed paths from $u$ to $v$
have the same length. 
\end{lemma}
\begin{proof} Let $Q$ be a solution to the transportation problem satisying the
equal-length property. We show that $D = Q$ is an equilibrium point. In any
connected component of $N$, fix the potential of an arbitrary vertex to zero and
then extend the potential function to the other vertices by the rule $\Delta_e
= L_e$. By the equal-length property, the potential function is well
defined. Let $Q'$ be the electrical flow induced by the potentials and edge
diameters. For any 
edge $e
= (u,v) \in N$, we have $Q'_e = D_e \Delta_e/L_e = D_e = Q_e$. 
For any edge $e
\not\in N$, we have $Q_e = 0 = D_e$. Thus, $D$ is an equilibrium point.

Let $D$ be an equilibrium point and let $Q_e$ be the corresponding current
along edge $e$. Whenever $D_e>0$, we have $\Delta_e = Q_e L_e / D_e = L_e$ because of the
equilibrium condition. Since all directed paths between any two vertices span the same
potential difference, $N$ satisfies the equal-length property. Moreover, by Lemma
\ref{lem:equilibrium-mincut-trasport}, $\min_S \capa(S)/b_s=1$, and hence, $Q$ is a
solution to the transportation problem with the equal-length property.  
\end{proof}

Let $\E$ be the set of equilibria and let $\Estar$ be the set of equilibria of
minimum cost. 

\begin{lemma}
\label{lem:monotonicity-sourcecut transportation}
Let $W = (\capa(\{s_0\})-1)^2$. Then, $\dot{W} = -2 W \le 0$ with equality iff 
$\capa(\{s_0\})=1$. 
\end{lemma}
\begin{proof} 
Let $C_0 = \capa(\{s_0\})$ for short. Then, since $\sum_{e \in \delta(\{s_0\})} \abs{Q_e}=1$,  
\[
\dot{W} = 2(C_0 -1) \sum_{e \in \delta(\{s_0\})} \left(\abs{Q_e} - D_e
\right) = 2(C_0 -1) (1 - C_0) = -2 (C_0 -1)^2 \le 0. \] \Lnegskip
\end{proof}

\noindent
The following functions play a crucial role: Let $F =\min_{S \in \calC} F_S$, and  
\begin{align*}
V_S &= \frac{1}{F_S} \sum_{e \in E} L_e D_e \text{ for each } S \in \calC, \\
V &= \max_{S \in \calC} V_S + W, \text{ and }\\
h &= - \frac{1}{F} \sum_{e \in E} R_e \abs{Q_e} D_e + \frac{1}{F^2} \sum_{e \in E} R_e D_e^2. 
\end{align*}


\begin{lemma}
\label{lem:equality transportation}
Let $S$ be a most constraining cut at time $t$. Then, $\dot{V}_S(t) \le
-h(t)$. 
\end{lemma}
\begin{proof}
Let $X$ be the characteristic vector of $\delta(S)$, that is, $X_e = 1$ if $e
\in \delta(S)$, and $X_e = 0$ otherwise. Observe that $F_S = F$ since $S$ is a most
constraining cut. Let $C = C_S$. 
We have
\begin{align*}
\dot{V}_S &= \sum_e \frac{\partial V_S}{\partial D_e} \dot{D}_e \\
&= \sum_e \frac{\abs{b_S}}{C^2} \left( L_e C - \sum_{e'} L_{e'} D_{e'} X_e \right)
\left( \abs{Q_e} - D_e \right) \\
&= \frac{\abs{b_S}}{C} \sum_e L_e \abs{Q_e} - \frac{\abs{b_S}}{C^2} \left(\sum_{e'} L_{e'} D_{e'}
\right) \left(\sum_e X_e \abs{Q_e} \right) + \\
& \qquad\qquad
- \frac{\abs{b_S}}{C} \sum_e L_e D_e  + \frac{\abs{b_S}}{C^2} \left(\sum_{e'} L_{e'} D_{e'} \right)
\left(\sum_e X_e D_e \right) \\
&\le \frac{\abs{b_S}}{C} \sum_e R_e \abs{Q_e} D_e  - \frac{b_S^2}{C^2} \sum_{e} R_{e} D_{e}^2
- \frac{\abs{b_S}}{C} \sum_e L_e D_e  + \frac{\abs{b_S}}{C} \sum_{e} L_{e} D_{e}  \\
&= -h. 
\end{align*}
The only inequality follows from $L_e=R_e D_e$ and $\sum_e X_e \abs{Q_e} \ge
\abs{b_S}$, which holds
because at least $b_S$ units of current must cross $S$.  
\end{proof}

\begin{lemma}
\label{lem:lyapunov transportation}
$\dot{V}$ exists almost everywhere. 
If $\dot{V}(t)$ exists, then $\dot{V}(t) \le -h(t) -2W(t) \le 0$, and
$\dot{V}(t)=0$ iff $\forall e,~\dot{D}_e(t)=0$.  
\end{lemma}
\begin{proof} The almost everywhere existence of $\dot{V}$ is shown as in
Lemma~\ref{lem:lyapunov}.

The fact that $W \ge 0$ is clear. We now show that $h \ge 0$. To this end, let
$f$ represent a solution to the (capacitated) transportation problem in an auxiliary network
having the same structure as $G$ and where the capacity of edge $e$ is set equal to
$D_e/F$; $f$ exists by Hoffman's circulation theorem \cite[Corollary 11.2g]{Schrijver:2003}: observe that for any cut $T$, $F_T \ge F$, and hence, $\abs{b_T} \le C_T/F$. 
Then, 
\begin{align*}
-h &= \frac{1}{F} \sum_e R_e \abs{Q_e} D_e  - \frac{1}{F^2} \sum_{e} R_{e} D_{e}^2 \\
&\le \frac{1}{F} \left(\sum_e R_e Q_e^2 \right)^{1/2} \left( \sum_e R_e D_e^2 \right)^{1/2} - \frac{1}{F^2} \sum_e R_e D_e^2 \\
&\le \frac{1}{F} \left( \sum_e R_e {f_e^2} \right)^{1/2} \left( \sum_e R_e D_e^2 \right)^{1/2} - \frac{1}{F^2} \sum_e R_e D_e^2 \\
&\le \frac{1}{F^2} \left( \sum_e R_e D_e^2 \right)^{1/2} \left( \sum_e R_e
D_e^2 \right)^{1/2} - \frac{1}{F^2} \sum_e R_e D_e^2 \\
&= 0, 
\end{align*}
where we used the following inequalities:
\begin{itemize}
\item[-] the Cauchy-Schwarz inequality $\sum_e (R_e^{1/2} \abs{Q_e}) (R_e^{1/2} D_e) \le (\sum_e R_e Q_e^2)^{1/2} (\sum_e R_e D_e^2)^{1/2}$;
\item[-] Thomson's Principle \eqref{eq:thomson} applied to the flows $Q$ and
$f$; $Q$ is a minimum energy flow solving the transportation problem, while $f$
is a feasible solution; and 
\item[-] the fact that $\abs{f_e} \le D_e/F$ for all $e \in E$.  
\end{itemize}

Finally, one can have $h=0$ if and only if all the above inequalities are
equalities, which implies that $\abs{Q_e}=\abs{f_e}=D_e/F$ for all $e$. And,
$W=0$ iff $\sum_{e \in \delta(\{s_0\})} D_e = 1 = \sum_{e \in \delta(\{s_0\})}
\abs{Q_e}$. So, $h=W=0$ iff $\abs{Q_e}=D_e$ for all $e$.   
\end{proof}


\begin{lemma}
\label{lem:h-uc transportation}
The function $t \mapsto h(t)$ is Lipschitz-continuous.
\end{lemma}
\begin{proof} The proof of Lemma~\ref{lem:h-uc} carries over. \end{proof}

\begin{lemma}\label{lem: Q minus D transportation} $\abs{D_e - \abs{Q_e}}$ converges to zero
for all $e \in E$. \end{lemma}
\begin{proof} The first and last paragraph of the proof of Lemma~\ref{lem: Q
minus D} carry over. We redo the second paragraph. 

The first paragraph establishes that for any $\varepsilon>0$, there is $t_0$ such
that $h(t) \le \varepsilon$ for all $t \ge t_0$.  Then, recalling that $R_e \ge
\Lmin/2$ for all sufficiently large $t$ (by Lemma \ref{lem:basic-properties}),
we find  
\begin{align*}
\sum_e \frac{\Lmin}{2} \left( \frac{D_e}{F} - \abs{Q_e} \right)^2 
&\le 
\sum_e R_e \left( \frac{D_e}{F} - \abs{Q_e} \right)^2 \\
&= \frac{1}{F^2} \sum_e R_e D_e^2 + \sum_e R_e Q_e^2 - \frac{2}{F} \sum_e R_e \abs{Q_e} D_e  \\
&\le \frac{2}{F^2} \sum_e R_e D_e^2 - \frac{2}{F} \sum_e R_e \abs{Q_e} D_e \\
&= 2 h \le 2 \varepsilon, 
\end{align*}
where we used once more the inequality $\sum_e R_e Q_e^2 \le \sum_e R_e
D_e^2/F^2$, which was proved in Lemma \ref{lem:lyapunov transportation}.  
This implies that for each $e$, $D_e/F - \abs{Q_e} \to 0$ as $t \to
\infty$. Summing across $e \in \delta(\{s_0\})$ and using Lemma
\ref{lem:basic-properties}(\ref{lem:source-flow}), we obtain $C_{\{s_0\}}/C - 1
\to 0$ as $t \to \infty$.  From Lemma \ref{lem:basic-properties}, $C_{\{s_0\}}
\to 1$ as $t \to \infty$, so $C \to 1$ as well. 
\end{proof}

We are now ready to prove that the set of equilibria is an attractor.

\begin{theorem}\label{thm: equilibria attract} 
The dynamics are attracted by the set $\E$ of equilibria. 
\end{theorem}
\begin{proof} Assume otherwise. Then, there is a network and initial conditions
for which the dynamics has an accumulation point $D$ that is not an
equilibrium; such an accumulation point exists because the dynamics are eventually confined to a compact set. 
Let $Q$ be the flow corresponding to $D$. Since $D$ is not
an equilibrium, there is an edge $e$ with $D_e \not= \abs{Q_e}$. This
contradicts the fact that $\abs{D_e - \abs{Q_e}}$ converges to zero for all
$e$.
\end{proof}

\begin{theorem}\label{thm:main transportation} If no two equilibria
have the same cost, the dynamics converge to a minimum
cost solution. 
\end{theorem} 
\begin{proof} Consider any equilibrium $D^*$, and let $Q^* = D^*$ be the
corresponding flow. Let $T^*$ be the edges carrying non-zero flow. $T^*$ must be
a forest, as otherwise, there would be two equilibria with the same
cost. Consider any edge $e = (u,v)$ of $T^*$, and let $S$ be the
connected component of $T^* \setminus e$ containing $u$. Then $Q^*_e = b(S)$, and hence,
distinct equilibria have distinct associated forests. We conclude that the set
of equilibria is finite. 

The $V$-value of $D^*$ is equal to the cost $\sum_e L_e Q^*_e$ of the
corresponding flow since $W = 0$ and $F = 1$ in an equilibrium. If no two
equilibria have the same cost, the $V$-values of distinct equilibria are
distinct. 

$V$ is a decreasing function and hence converges. Since the dynamics are
attracted to the set of equilibria, $V$ must converge to the cost of an
equilibrium. Since the equilibria are a discrete set, the dynamics must
converge to some equilibrium. Call it $D^*$. 

We next show that $D^*$ is a minimum cost solution to the transportation
problem. Orient the edges in the direction of the flow $Q^*$. If $Q^*$ is not a
minimum cost flow, there is an oriented path $P$ from a supply node $u$ to a demand
node $v$ such that $Q_e > 0$ for all edges of $P$, and $P$ is not a shortest
path from $u$ to $v$. The potential difference $\Delta_{uv}$ converges to
$L_P$. We now derive a contradiction as in the proof of
Lemma~\ref{lem:convergence of Delta}. 

Let $P'$ be a shortest path from $u$ to $v$ in $G$, let $L^* = L_{P'}$ be its
length, and let 
$W_{P'} = \sum_{e \in P'} L_e \ln D_e$. 
We have
\[ 
 \dot{W}_{P'} = \sum_{e \in P'} \frac{L_e}{D_e} (\abs{Q_e} - D_e)  = \sum_{e
\in P'} \abs{\Delta_e} - \sum_{e \in P'} L_e \ge p_{u}-p_{v} - L_{P'} =
\Delta_{uv} - L^*.\]
Let $\delta > 0$ be such that there is no path from $u$ to $v$ with length in the
open interval $(L^*, L^* + 2 \delta)$. Then, $\Delta - L^*\ge \delta$
for all sufficently large $t$, and hence, $\dot{W}_{P'} \ge \delta$ for all
sufficiently large $t$. Thus, $W_{P'}$ goes to $+\infty$. However, $W_{P'} \le
n \Lmax$ for all sufficiently large $t$ since $D_e \le 2$ for all $e$ and $t$
large enough. This is a contradiction. 
\end{proof}

\begin{lemma}\label{lem: minimizing V(D) is the same as transportation problem}
The problem of minimizing $V(D)$ for $D \in \RR^E_+$ is equivalent to the
transportation problem. 
\end{lemma}
\begin{proof}
By introducing an additional variable $F = \min_{S} \capa(S)/\abs{b(S)} > 0$, the problem of minimizing $V(D)$ is equivalently formulated as
\begin{align*}
\min\, & \frac{1}{F} \sum_e L_e D_e + \left(\sum_{e \in \delta(\{\source\})} D_e - 1 \right)^2 \\
\text{s.t. } & \capa(S)/\abs{b(S)} \ge F \qquad \forall S \in \calC \\
& F >0 \\
& D \ge 0. 
\end{align*}
Substituting $x_e = D_e / F$, we obtain
\begin{align*}
\min\, & \sum_e L_e x_e  + F^{1/2} \left(\sum_{e \in \delta(\{\source\})} x_e - \frac{1}{F} \right)^2 \\
\text{s.t. } & \sum_{e \in \delta(S)} x_e \ge \abs{b(S)} \qquad \forall S \in \calC \\
& x \ge 0, F > 0
\end{align*}
which is easily seen to be equivalent to the (fractional) transportation problem. 
\end{proof}

\bibliographystyle{alpha}
\bibliography{ref}



\end{document}